\documentclass[journal]{IEEEtran}
\ifCLASSINFOpdf
\usepackage[pdftex]{graphicx}
\else
\usepackage[dvips]{graphicx}
\usepackage[caption=false,font=footnotesize]{subfig}
\fi
\usepackage[caption=false,font=footnotesize]{subfig}
\usepackage{setspace}
\usepackage{algorithm}
\usepackage{algorithmicx} 

\usepackage{multicol}
\usepackage{lipsum}
\usepackage{setspace}
\usepackage{arydshln}
\usepackage{makecell}
\usepackage{rotating}
\usepackage{makecell, multirow}
\usepackage[skip=1ex]{caption}

\usepackage{fix-cm}
\usepackage{xcolor,soul,framed} 
\usepackage{subfiles}
\colorlet{shadecolor}{yellow}
\usepackage[pdftex]{graphicx}
\graphicspath{{../pdf/}{../jpeg/}}
\DeclareGraphicsExtensions{.pdf,.jpeg,.png}
\usepackage{tikz}
\usepackage{array}
\usepackage{mdwmath}
\usepackage{mdwtab}
\usepackage{eqparbox}
\usepackage{url}
\usepackage{cite}
\usepackage{epsfig,amsmath,amssymb,epsf,amsthm,scalefnt,multirow,subfig}
\usepackage{xcolor}
\usepackage{float}
\usepackage{psfrag}
\usepackage{algpseudocode}
\usepackage{verbatim}
\usepackage{tikz}
\usepackage{makecell}
\usetikzlibrary{tikzmark,calc,decorations.pathreplacing}
\usepackage{amsmath}
\usepackage[utf8]{inputenc}

\algnewcommand{\LeftComment}[1]{\Statex \(\triangleright\) #1}

\newtheorem{theorem}{Theorem}

\newtheorem*{lemma*}{Lemma}
\newtheorem{remark}{Remark}

\newtheorem{corollary}{Corollary}

\newtheorem{proposition}{Proposition}

 \def\bomega  
\def\bTheta{{\pmb{\Theta}}}

\def\b0{{\pmb{0}}} 

\allowdisplaybreaks[4]

\newcolumntype{P}[1]{>{\centering\arraybackslash}p{#1}}
\newcommand{\sh}[1]{\textcolor{black}{{#1}}}

\newcommand{\ky}[1]{\textcolor{orange}{{#1}}}

\begin{document}
\title{Trajectory Optimization for Cellular-Enabled UAV with Connectivity and Battery Constraints}

\author{Hyeon-Seong Im,~
        Kyu-Yeong Kim,~\IEEEmembership{Graduate Student Member,~IEEE,}
        and Si-Hyeon Lee,~\IEEEmembership{Senior Member,~IEEE}
\thanks{
Copyright (c) 2025 IEEE. Personal use of this material is permitted. However, permission to use this material for any other purposes must be obtained from the IEEE by sending a request to pubs-permissions@ieee.org.

This article was presented in part at the IEEE Vehicular Technology Conference (VTC) 2023-Fall \cite{Im:2023_VTC}.
This work was supported in part by the Institute of Information \& Communications Technology Planning \& Evaluation (IITP) grant (No.RS-2024-00360387, Development of core security technology utilizing multimodal properties of wireless communication channels, contribution rate: 50\%), in part by the IITP-ITRC (Information Technology Research Center) grant (IITP-2025-RS-2020-II201787, Development of communication/computing-integrated revolutionary technologies for superintelligent services, contribution rate: 30\%), and in part by Unmanned Vehicles Core Technology Research and Development Program through the National Research  Foundation of Korea (NRF) and the Unmanned Vehicle Advanced Research Center (UVARC) grant (No.2020M3C1C1A01084524, Development of anti-jamming for unmanned vehicles security and detection and counter technology to unlicensed unmanned vehicles, contribution rate: 20\%), all funded by the Korea government (MSIT).

H.-S. Im is with LIG Nex1, Seongnam 13488, South Korea (e-mail: hyeonseong.im@lignex1.com).  He was with the School of Electrical Engineering,  Korea Advanced Institute of Science and Technology (KAIST), Daejeon 34141, South Korea, when conducting this work. K.-Y. Kim and S.-H. Lee (Corresponding Author) are with the School of Electrical Engineering, KAIST, Daejeon 34141, South Korea (e-mail: kimyou283@kaist.ac.kr, sihyeon@kaist.ac.kr). }
}

\maketitle

\begin{abstract}
We address the path planning problem for a cellular-enabled unmanned aerial vehicle (UAV) considering both connectivity and battery constraints. The UAV's mission is to expeditiously transport a payload from an initial point to a final point, while persistently keeping the connection with a base station and complying with its battery limit. At a charging station, the UAV's depleted battery can be swapped with a completely charged one. 
Our primary contribution lies in proposing an algorithm that outputs an optimal UAV trajectory with polynomial computational complexity, by converting the problem into an equivalent two-level graph-theoretic shortest path search problem. 
We compare our algorithm with several existing algorithms with respect to performance and computational complexity, and show that only our algorithm outputs an optimal UAV trajectory in polynomial time. 
Furthermore, we consider other objectives of minimizing the UAV energy consumption and of maximizing the deliverable payload weight, and propose algorithms that output an optimal UAV trajectory in polynomial time.

\end{abstract}


%
\IEEEpeerreviewmaketitle


\section{Introduction}\label{sec1}
Unmanned aerial vehicles (UAVs) are extensively applied across various scenarios, such as delivery and transportation \cite{Zhang:2021_2}, aerial surveillance and monitoring \cite{Kanistras:2013}, flying base stations (BSs) \cite{Mozaffari:2019}, and data collection and/or power transfer for IoT devices  \cite{Yu:2021}, due to their high mobility, free movement, and cost-effectiveness  \cite{Mozaffari:2019,Shi:2018}. It has been actively studied to design the UAV trajectory according to each operational scenario. For UAV-aided communication scenarios, the UAV trajectory has been optimized taking into account various factors, e.g., minimizing energy while satisfying throughput demands  \cite{Zeng:2019,Qi:2020}, \sh{minimizing the age of information (AoI) of IoT devices \cite{Yi:2023}, jointly optimizing the transmission rate and the rate variation in aerial video streaming scenario \cite{Zhan:2024_2}}, and improving secrecy rate in the presence of an eavesdropper \cite{Li:2019,Cui:2018}. For delivery or transportation scenarios, it is utmost important  to swiftly and safely transport the given objects to their desired destinations. Thus, for such scenarios, the problem of designing UAV trajectory has been formulated as minimizing the mission time or \sh{the energy consumption (equivalent to minimizing the mission time of UAV with fixed speed)} with some constraints such as \sh{connectivity \cite{Zhang:2019,Zhang:2019_2,Chen:2020,Zhan:2022,Zhan:2022_2,Zhang:2021,Esrafilian:2020,Chapnevis:2021,Zeng:2019_2,Khamidehi:2020,Wang:2022,Chen:2022}}, restricted airspace \cite{Khamidehi:2020,Wang:2022}, collision avoidance between UAVs \cite{Wang:2022,Chen:2022}, and battery limits \cite{Sundar:2014,Coelho:2017,Fan:2023,Arafat:2022}. 
In particular, it is important for such scenarios to consistently keep the connection between the UAV and the control station. This persistent connectivity is 
essential for tracking the real-time location of the cargo in delivery scenarios \cite{Zhang:2019} and for providing real-time communication service to passengers in transportation scenarios like urban air mobility (UAM) \cite{Cohen:2021}. It is also vital when the manual control is necessary for the UAV to evade unexpected adverse weather conditions or avoid collisions with aerial obstacles \cite{Banafaa:2024}. However, maintaining direct connection with the control station becomes challenging if the initial and the final points are far away, due to low received signal strength by low line-of-sight (LoS) probability and long communication distance. Cellular-enabled UAV communication is a potential approach for this problem \cite{Zhang:2019}, wherein the UAV communicates with its control station by connecting with a close BS and the underlying cellular network \cite{Agyapong:2014}. 

Our paper addresses the path planning problem for UAVs performing delivery or transportation missions, jointly considering connectivity with the cellular network and UAV’s battery limit. In the absence of battery limit, the problem of minimizing the mission time \sh{(or the energy consumption)} while ensuring continuous connection with the cellular network has been \sh{extensively studied \cite{Zhang:2019,Zhang:2019_2,Chen:2020,Zhan:2022,Zhan:2022_2,Zhang:2021,Esrafilian:2020,Chapnevis:2021,Zeng:2019_2,Khamidehi:2020,Chen:2022,Wang:2022}.}  
The authors of \cite{Zhang:2019} concentrated on planning an optimal path between a initial and a final points while ensuring continuous connection with a BS. To simplify the problem, they assumed that the UAV and a BS can be connected if their communication distance is not larger than a certain threshold. Then, they transformed the problem into graph-theoretic path finding and convex optimization problems, and proposed an optimal algorithm with non-polynomial time complexity and two sub-optimal algorithms with polynomial time complexity. 
The study of characterizing an optimal path under the connectivity constraint has been extended in various directions, e.g., \sh{allowing a certain duration or ratio of communication outage \cite{Zhang:2019_2,Chen:2020,Zhan:2022}, assuming the knowledge of the radio map \cite{Zhan:2022_2,Zhang:2021}}, considering 3-dimensional (3D) space \cite{Esrafilian:2020,Zhang:2021},
and considering the collaboration of multiple UAVs \cite{Chapnevis:2021}. 
Specifically, the work \cite{Chen:2020} proposed an intersection method, significantly reducing the computational complexity by converting the problem into a graph-theoretic path finding problem whose vertex set consists of the intersection points of the coverage boundaries of the BSs. Moreover, the work \cite{Zhang:2021} also used a graph-theoretic approach even for 3D path finding problem with a realistic communication environment considering signal blockage and reflection by buildings and interference from other BSs, by quantizing the radio map to finite grid points. On the other hand, reinforcement learning (RL) \cite{Sutton:2018} based approaches may be effective for scenarios that the UAV only has limited prior knowledge about communication and transportation environment, since it can empirically learn the environment. Several studies have explored to learn effective UAV paths by applying \sh{RL-based approaches \cite{Zhan:2022,Zhan:2022_2,Zeng:2019_2,Khamidehi:2020,Chen:2022,Wang:2022}.} However, It is important to acknowledge that the RL-based approach may not consistently output an optimal path, and the training process involved in RL may demand considerable computing time and resources. 

In practice, it is important to consider the limited battery capacity of the UAV. There have been a few works on designing UAV trajectory performing delivery or transportation missions taking into account the limited battery capacity \cite{Sundar:2014,Coelho:2017,Fan:2023,Arafat:2022}. The work \cite{Sundar:2014} considered a variant of the travelling salesman problem (TSP) that aims to derive a shortest route visiting each target node once, while considering the limited battery capacity of the UAV and charging stations to replenish its energy. Such a UAV route optimization problem with  TSP formulation taking into account the battery limit has been extended by considering multiple UAVs \cite{Coelho:2017,Fan:2023} and grouping target nodes into clusters \cite{Arafat:2022}.  However, the problem of UAV path planning to minimize the mission time for delivery or transportation scenarios while considering both the connectivity and the battery replenishment has not been well studied. We note that for other UAV utilization scenarios, such a problem of  UAV path planning taking into account both the connectivity and the battery replenishment has been studied, but for different objectives depending on the assigned missions, e.g., minimizing the energy consumption of UAV \cite{Zeng:2019,Qi:2020} or the AoI of devices \cite{Yi:2023,Zhan:2024}. 

Our primary contribution involves proposing a path planning algorithm that outputs an optimal flight path of a cellular-enabled UAV in polynomial time complexity, to expeditiously transport a payload from an initial point to a final point, while persistently keeping the connection with a BS and complying with its battery limit. It is assumed that the UAV can establish a connection with a BS if they are closer than a certain threshold similarly as in \cite{Zhang:2019}, but we allow that the threshold can be different for each BS due to interference from other BSs. The UAV’s depleted battery can be swapped with another completely charged one at a charging station, which may involve a certain delay depending on waiting and replacing times \cite{Lee:2015}. The contributions of this paper are summarized as follows:

\begin{itemize}
\item Our problem of optimizing UAV trajectory involves determining the UAV's path, speed, and the order of visiting charging stations. We solve this problem by first reformulating the problem as a two-level graph-theoretic shortest path search problem, and then applying Dijkstra algorithm \cite{Dijkstra:1959}. More precisely, we initially search an optimal path and the corresponding maximum allowable speed for traveling between each pair of charging stations  without swapping to a fully charged battery. Subsequently, we ascertain an optimal sequence of visiting charging stations. To validate the efficacy of our approach, we analytically and numerically compare our algorithm with existing algorithms in \cite{Zhang:2019,Chen:2020} with marginal modifications, and show that only our algorithm yields an optimal solution in polynomial time.

\item Characterizing the maximum deliverable payload weight under the connectivity and battery constraints is another interesting problem. We propose a graph-theoretic  algorithm that yields an optimal solution to this problem in polynomial time. It first transforms the delivery environment into a weighted graph and finds the longest connectivity-critical edge between the initial and the final points in the graph. Then, it derives the largest payload weight which can be delivered over the edge without replacing the battery.

\end{itemize}
The remaining of this paper is organized as follows. In Section \ref{sec2}, we present the system model and formulate the optimization problem of finding the fastest UAV route under the connectivity and battery constraints. Our proposed algorithms that output optimal UAV trajectories without and with the battery limit are presented in Sections \ref{sec3} and \ref{sec4}, respectively. In Section \ref{sec5}, other objectives of minimizing the UAV energy consumption  and of maximizing the deliverable payload weight are considered and the corresponding optimal algorithms are presented. We provide various numerical results in Section \ref{sec6}. Finally, the paper is concluded in Section \ref{sec7}.

\section{Problem Statement}\label{sec2}
We consider a cellular network with $M$ BSs and $N\leq M$ charging stations (CSs). In this network, a UAV with limited battery capacity travels from an initial point $\mathbf{U}_0$ to a final point $\mathbf{U}_F$ to deliver or transport a payload. The detailed description of the UAV model is in Section \ref{sec2A}. The UAV should maintain the connectivity with one of the BSs while delivering the payload. The BS model and the BS-UAV connectivity are described in Section \ref{sec2B}. The UAV can swap its battery at a CS if needed, as explained in Section \ref{sec2C}. This paper's goal is to plan an optimal transportation path from $\mathbf{U}_0$ to $\mathbf{U}_F$, which minimizes the mission (travel) time $T$ encompassing both the flight time and the battery swapping time at CSs. This optimization problem is formally presented in Section \ref{sec2D}. The overall model is illustrated in Fig. \ref{Fig1}.

\begin{figure}
\centering
\includegraphics[width=0.95\columnwidth]{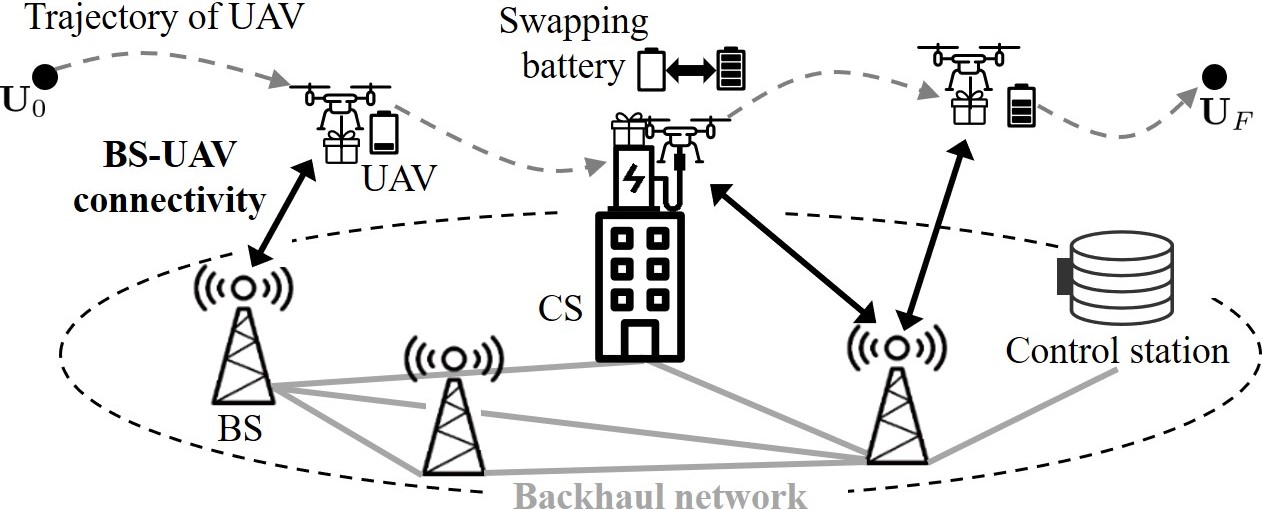}
\caption{An example of delivery scenario of a cellular-enabled UAV. The UAV delivers a payload from  $\mathbf{U}_0$ to $\mathbf{U}_F$ under limited battery constraint while communicating with a BS, where all the BSs and CSs are connected to the control station through the backhaul network. In the CS, the UAV can swap its battery.}\label{Fig1}
\end{figure}

\subsection{UAV Model}\label{sec2A}
A UAV performs a mission to transport a payload from an initial point $\mathbf{U}_0$ to a final point $\mathbf{U}_F$ within the framework of the cellular network, under the assumption that the UAV is equipped with rotary-wings and travels at a fixed altitude $H\in[H_\mathrm{min},H_\mathrm{max}]$, where $H_\mathrm{min}$ is determined by the heights of obstacles in the network and $H_\mathrm{max}$ corresponds to the maximum allowable altitude according to government regulations. The 3-dimensional (3D) coordinates of $\mathbf{U}_0$, $\mathbf{U}_F$, and the position of the UAV at time $t$ are represented by $(x_0,y_0,H)$, $(x_F,y_F,H)$, and $(x(t),y(t),H)$, respectively, where those 2-dimensional (2D) coordinates (i.e., horizontally projected locations) are indicated as $\mathbf{u}_0=(x_0,y_0)$, $\mathbf{u}_F=(x_F,y_F)$, and $\mathbf{u}(t)=(x(t),y(t))$, respectively. While performing the mission, the UAV can adaptively adjust its speed $v(t)\triangleq\|\nabla_t\mathbf{u}(t)\|$ at time $t$ in finite set $\mathcal{V}=\{0,v_1,v_2,...,v_q\}$ ($0<v_1<...<v_q$).
In terms of energy usage, our paper concentrates solely on the energy consumed by the UAV propulsion, by the relatively marginal impact of communication power consumption \cite{Mozaffari:2019}. The overall UAV consists of the UAV body, a battery, and its payload, where those weights are denoted by $w_1$, $w_2$, and $w_3$, respectively with the total UAV weight $w\triangleq w_1+w_2+w_3$.
The power consumed by propulsion (in Watts) during flight at speed $v$ is
\begin{align}
\begin{split}\label{eq:1}
&P_\mathrm{UAV}(v)=P_1\left(1+{3({v}/{v_\mathrm{tip}}})^2\right)+P_2(w)\cdot\\
&\big(\sqrt{1\!+\!0.25{({v}/{v_0(w)}})^4}\!-\!{0.5({v}/{v_0(w)})^2}\big)^{0.5} \!\!\!\!\! +0.5S_\mathrm{FP}{\rho} v^3,
\end{split}
\end{align}
where $v_\mathrm{tip}$, $v_0(w)$, $P_1$, $P_2(w)$, $S_\mathrm{FP}$, and $\rho$ refer to the parameters established through the UAV's physical structure and delivery environment, with pointing out that only $v_0(w)$ and $P_2(w)$ rely on the total weight $w$. A comprehensive discussion regarding the propulsion power consumption \eqref{eq:1} including its parameters will be provided in Section \ref{sec6}. The UAV receives the energy for operation from its battery. The capacity (in Joules) of the battery is directly proportional to its weight $w_2$ as follows \cite{Zhang:2021_2}:
\begin{align}
C_\mathrm{batt}=\epsilon_\mathrm{batt}w_2,\label{eq:2}
\end{align}
where $\epsilon_\mathrm{batt}$ means the energy density (per unit weight) of fully charged battery. From \eqref{eq:1} and \eqref{eq:2}, the maximum flight distance under the assumption that UAV does not swap its battery and maintains a constant speed $v$ is provided as follows \cite{Zhang:2021_2}:
\begin{align}
d_\mathrm{fly}(v)=v\cdot{{\gamma\eta C_\mathrm{batt}}/({r_\mathrm{safe}P_\mathrm{UAV}(v)})},\label{eq:3}
\end{align}
where $\gamma\in(0,1)$ represents the depth of discharge, i.e., the maximum fraction of the energy that can be used in the fully charged battery, $\eta\in(0,1)$ signifies the ratio of the charged energy transferable to the UAV body under the circuit resistance, and $r_\mathrm{safe}\in(1,\infty)$ denotes the energy reserving factor of the battery for unforeseeable situations like strong winds. More precisely, the maximum flight distance $d_\mathrm{fly}(v)$ is obtained by first dividing the actual maximum available energy from the complete charged battery, ${\gamma C_\mathrm{batt}}\over r_\mathrm{safe}$ (in Joules) into the consumed power in the UAV body, ${P_\mathrm{UAV}(v)}\over{\eta}$ (in Watts) and then multiplying the speed $v$.


\subsection{BS-UAV Connectivity}\label{sec2B}
In the cellular network, the UAV can communicate with one of $M\geq 1$ BSs, where the $m$th BS for $m\in[1:M]$ is denoted by $\mathrm{BS}_m$. The 3D coordinate of $\mathrm{BS}_m$ is $(a_{m1},a_{m2},H_\mathrm{BS})$, with its 2D coordinate $\mathbf{a}_m=(a_{m1},a_{m2})$. It is assumed that the altitudes of all BSs are identical to $H_\mathrm{BS}<H$. Each BS is equipped with one omni-directional antenna and operates at the same transmission power denoted by $P_\mathrm{tx}$. We note that there exists a control station to plan the UAV trajectory and approve the handover between BSs. The control station is connected to every BS through a backhaul network as shown in Fig. \ref{Fig1} and hence it is assumed that the control station can maintain the control of the UAV if the communication rate between the UAV and its connected BS is not smaller than a certain threshold.

We assume that the channel between the UAV and a BS is determined by the line-of-sight (LoS) probabilistic model \cite{Al-Hourani:2014}. The LoS probability between the UAV and $\mathrm{BS}_m$ at time $t$, $p_m(t)\in[0,1]$ is given as follows \cite{Al-Hourani:2014}:
\begin{align}
p_m(t)={1/({1+\mu_1\exp(-\mu_2(\theta_m(t)-\mu_1))}}),\label{eq:3.1}
\end{align}
where $\mu_1>0$ and $\mu_2>0$ refer to the parameters to determine the LoS probability and $\theta_m(t)$ is the elevation angle between the UAV and $\mathrm{BS}_m$ at time $t$. We note that the LoS probability $p_m(t)$ increases as the elevation angle $\theta_m(t)$ increases.
The expected path loss between the UAV and $\mathrm{BS}_m$ at time $t$, $\Lambda_m(t)$ (in $\mathrm{dB}$) is given as follows \cite{Al-Hourani:2014_2}:
\begin{align}
\Lambda_m(t) = \mathrm{FSPL}_m(t)+p_m(t)\cdot\zeta_1 + (1-p_m(t))\cdot\zeta_2,\label{eq:4}
\end{align}
where $\mathrm{FSPL}_m(t)$ is the free space path loss (FSPL) which only depends on the distance between the UAV and $\mathrm{BS}_m$ and $\zeta_1>0$ and $\zeta_2>\zeta_1$ refer to the excessive path losses for LoS and non-LoS (NLoS) links, respectively \cite{Al-Hourani:2014_2}. The received signal to interference plus noise ratio (SINR) from $\mathrm{BS}_m$ to the UAV at time $t$ is
\begin{align}
\mathrm{SINR}_m(t)={{P_\mathrm{tx}\cdot 10^{\Lambda_m(t)/10}}/(\!\!\!\!\!{\sum_{m'\in[1:M]\setminus m}\!\!\!\!\!I_{m'm}(t)+N_0})},\label{eq:5}
\end{align}
where $I_{m'm}(t)$ is the interference power by $\mathrm{BS}_{m'}$ at time $t$ when the UAV is communicating with $\mathrm{BS}_{m}$ and $N_0$ is the additive noise power.\footnote{Our path loss model is based on large-scale fading, i.e., small-scale fading effects are ignored. However, we note that our results are also applicable under small-scale fading  by considering the SINR averaged over the randomness.}  Note that $I_{m'm}$ would be equal to zero if $\mathrm{BS}_{m'}$ uses a different frequency band from $\mathrm{BS}_{m}$, 
and even if the two BSs use the same frequency band, it will become negligible if $\mathrm{BS}_{m'}$ is far away from the UAV.

To maintain the control of the UAV, the communication rate from a BS to the UAV should not be less than the minimum required data rate, i.e., \sh{the SINR from $\mathrm{BS}_m$ to the UAV, maximized over $m\in[1:M]$,} should satisfy
\begin{align}
\max_{m\in[1:M]} \mathrm{SINR}_m(t)\geq \mathrm{SINR}_\mathrm{th}\label{eq:6}
\end{align}
for any time $t$ where $\mathrm{SINR}_\mathrm{th}$ is the hard SINR threshold to achieve the minimum required data rate. We note that the connectivity constraint \eqref{eq:6} does not mandatorily require LoS links between the UAV and BSs since the SINR value \eqref{eq:5} from a BS to the UAV is determined by averaging the pathloss over the LoS probability between them. In weak interference regime, i.e., the frequency reuse factor is sufficiently low, it can be easily checked that the condition \eqref{eq:6} can be equivalently written as  $\min_{m\in[1:M]}\|\mathbf{u}(t)-\mathbf{a}_m\|\leq d_0$ for some $d_0$, where we call $d_0$ the base coverage radius of each BS.
For other cases, however, it is generally hard to represent the exact coverage region satisfying  \eqref{eq:6} in a simple form. For tractable analysis, we introduce the coverage offset $\lambda_m\in [0,d_0]$ for $\mathrm{BS}_m$ and assume that the UAV can connect with $\mathrm{BS}_m$  if the UAV is in the effective coverage region of $\mathrm{BS}_m$ given as $\|\mathbf{u}(t)-\mathbf{a}_m\|\leq d_0-\lambda_m$, as explained in more detail in Remark \ref{rmk1}. 
Consequently, by introducing offsets $\lambda_m$ taking into account the effect of interference, we assume that \eqref{eq:6} holds if the following equation holds:
\begin{align}
\min_{m\in[1:M]}\|\mathbf{u}(t)-\mathbf{a}_m\|+\lambda_m\leq d_0.\label{eq:7}
\end{align}
Note that the connectivity condition is stated for downlink communications, but it can be defined similarly as in \eqref{eq:6} and \eqref{eq:7} for uplink communications.

\begin{remark}\label{rmk1}
Note that the coverage offset $\lambda_m$ for $m\in[1:M]$ depends on the environment around $\mathrm{BS}_m$. More specifically, $\lambda_m$ is chosen in a way that the UAV can connect with $\mathrm{BS}_m$ as long as its location $\mathbf{u}(t)$ satisfies $\|\mathbf{u}(t)-\mathbf{a}_m\|\leq d_0-\lambda_m$.  Fig. \ref{Fig2} illustrates an example of choosing $\lambda_m$ in the presence of interference. In general, we have large (small) $\lambda_m$ for urban (suburban) environments since there are many (few) other BSs around $\mathrm{BS}_m$. A higher flight altitude of the UAV induces a larger coverage offset $\lambda_m$ since the LOS probability of the channel between the UAV and other BSs increases  \cite{Lin:2019}. Also, $\lambda_m$ increases in the traffic of other BSs around $\mathrm{BS}_m$ \cite{Zhang:2021}.
\end{remark}

\begin{figure}
\centering
\includegraphics[width=0.8\columnwidth]{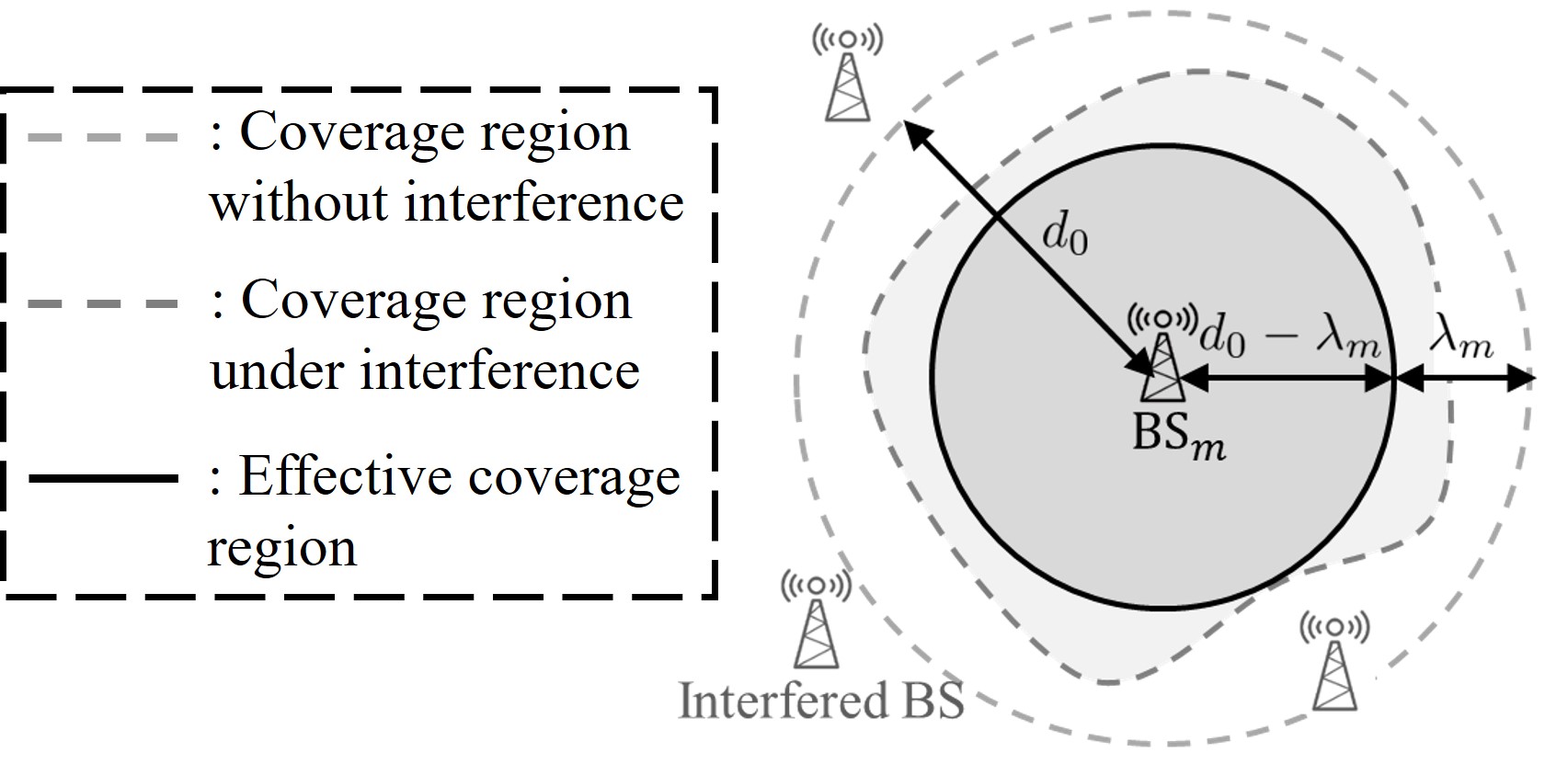}
\caption{An example of the effective coverage region of $\mathrm{BS}_m$ in the presence of interference.}\label{Fig2}
\end{figure}

\subsection{Charging Station Model}\label{sec2C}
To successfully transport the payload in the situation that the initial and the final points are far away and the battery capacity is limited, we assume that the UAV’s depleted battery can be swapped with a completely charged one at a charging station. \sh{It is assumed that all the CSs have the same height of $H_\mathrm{CS}\leq H$.} The $n$th charging station for $n\in[1:N]$ is denoted by $C_n$, and its 3D coordinate is indicated by $(c_{n1},c_{n2},H_\mathrm{CS})$ with the 2D coordinate $\mathbf{c}_n=(c_{n1},c_{n2})$, where $\mathbf{c}_n\neq \mathbf{c}_{n'}$ if $n\neq n'$. 
To expedite the battery replacement, we assume that every CS employs an automated battery swapping system, as described in \cite{Lee:2015}. 
The total delay for battery replacement at CS $C_n$, which is denoted as $\tau_{C_n}$ and bounded by the interval $[0,\tau_\mathrm{max}]$, encompasses waiting and battery swapping times (e.g., the automated battery swapping system in \cite{Lee:2015} takes about $60$ seconds for the entire battery swapping process) and the additional penalty for UAV takeoff and landing at the CS. It is worth mentioning that the waiting time can differ in accordance with the congestion level of the charging station. Hence, the delay $\tau_{C_n}$ is contingent upon the CS index $n$. We assume that all CSs are connected to the control station via backhaul network as illustrated in Fig. \ref{Fig1}, which allows the control station to have the complete knowledge of the delay for each CS and enables it to plan an effective UAV trajectory.

\subsection{Goal}\label{sec2D}
Our goal is to plan an optimal transportation path from $\mathbf{U}_0$ to $\mathbf{U}_F$, which minimizes the mission time $T$ that encompasses both the airborne flight time and the delay for battery replacement at CSs. The formulation of the optimization problem is represented as follows:
\begin{align} 
&\textbf{Problem 1} \cr
&\mbox{Objective:~}~~~~ \min_{T\geq 0,\{\mathbf{u}(t),\ \psi(t),\ t\in[0,T]\}} T\label{eq:8}\\
&\mbox{Constraints: }\cr
&\mathbf{u}(0)=\mathbf{u}_0,\ E_\mathrm{batt}(0)=C_\mathrm{batt},\ 
 \mathbf{u}(T)=\mathbf{u}_F \label{eq:9}\\  
&\mathbf{u}(t)\in\mathbb{R}^2,\ v(t)\in \mathcal{V},\ \psi(t)\in[0:N],\ t\in[0,T] \label{eq:9.1}\\
&\min_{m\in[1:M]}\|\mathbf{u}(t)-\mathbf{a}_m\|+\lambda_m\leq d_0,\ t\in[0,T]\label{eq:11}\\
&\psi(t)=0\ \mathrm{if~} \mathbf{u}(t)\not\in\{\mathbf{c}_n|n\in[1:N]\},\ t\in[0,T]\label{eq:12}\\
&\psi(t)\in\{0,n\} \ \mathrm{if~}  \mathbf{u}(t)=\mathbf{c}_n,\ n\in[1:N],\ t\in[0,T]\label{eq:13}\\
&E_\mathrm{batt}(t)\geq (1-(\gamma/ r_\mathrm{safe}))\cdot C_\mathrm{batt},\ t\in[0,T]\label{eq:14}\\
&-\nabla_t  E_\mathrm{batt}(t)=\! P_\mathrm{UAV}(v(t))/\eta\ \mathrm{if~} \psi(t)=0,\ t\in[0,T]\label{eq:15}\\
&-\nabla_t E_\mathrm{batt}(t)=\! 0\ \mathrm{if~} \psi(t)\in [1:N],\ t\in[0,T]\label{eq:16}\\
&E_\mathrm{batt}(t)=C_\mathrm{batt}\ \mathrm{if~} \psi(t)\in[1:N] \text{ and }\cr
&~~t-\max_{t_1}\{t_1|\psi(t_1)=0,t_1\in [0,t]\}=\tau_{C_{\psi(t)}},\ \! t\in[0,T]\label{eq:17}
\end{align}
where the auxiliary indicator $\psi(t)\in[0:N]$ is used to distinguish whether the UAV is airborne $(\psi(t)=0)$ or stays at CS $C_n$ $(\psi(t)=n)$ at time $t$ and $E_\mathrm{batt}(t)\geq 0$ refers to the battery's remaining energy at time $t$. Here, \eqref{eq:9} means that the UAV departs from $\mathbf{u}_0$ with fully charged battery and arrives at $\mathbf{u}_F$ at time $T$, \eqref{eq:9.1} corresponds to the range of optimizing variables including the set $\mathcal{V}$ of possible UAV speeds, \eqref{eq:11} is the connectivity constraint in \eqref{eq:7}, \eqref{eq:12}-\eqref{eq:13} identifies whether the UAV is airborne or stays at a CS, and \eqref{eq:14} means that the actual maximum available energy in the fully charged battery is ${\gamma C_\mathrm{batt}}\over r_\mathrm{safe}$. Next, \eqref{eq:15} and \eqref{eq:16} means the power usage during flight and battery replacement at a CS, respectively, and \eqref{eq:17} signifies that the depleted battery is just swapped to fully charged one.

We point out that Problem 1 doesn't fall under convex optimization due to the variable $\psi(t)\in[1:N]$ with discrete codomain and the non-convex constraint \eqref{eq:11}. Furthermore, optimizing $\mathbf{u}(t)$ and $\psi(t)$ over continuous $t\in[0,T]$ adds to the non-triviality of the problem. To address the challenges, in Sections \ref{sec3} and \ref{sec4}, we initially reframe the problem within a weighted graph methodology. Subsequently, we analytically show that leveraging the Dijkstra algorithm \cite{Dijkstra:1959} facilitates solving the problem efficiently in polynomial time.

\section{Optimal Trajectory with the Connectivity Constraint}\label{sec3}

This section presents an optimal method to solve Problem 1 assuming unlimited battery capacity (i.e., battery constraint is not considered and UAV does not visit any CSs). Without the battery limit, it is observed that the maximum flying speed $v_q$ minimizes the travel time from $\mathbf{U}_0$ to $\mathbf{U}_F$ and hence we assume that the speed is fixed at $v(t)=v_q$ for $t\in[0,T]$. This optimization problem can be reformulated from Problem 1 as follows:
\begin{align} 
&\textbf{Problem 1-1} \cr
&\mbox{Objective: }~~\min_{T\geq 0,\{\mathbf{u}(t),\ t\in[0,T]\}} T\label{eq:19}\\
&\mbox{Constraints:} \cr 
&\mathbf{u}(0)=\mathbf{u}_0,\ \mathbf{u}(T)=\mathbf{u}_F\label{eq:20}\\  
&\mathbf{u}(t)\in\mathbb{R}^2,\ v(t)=v_q,\ t\in[0,T]\label{eq:21}\\
&\min_{m\in[1:M]}\|\mathbf{u}(t)-\mathbf{a}_m\|+\lambda_m\leq d_0,\ t\in[0,T]\label{eq:22}
\end{align}
We point out that the problem still remains challenging due to the non-convex constraint \eqref{eq:22} and infinite number of control variables. 

To tackle Problem 1-1, we introduce a novel algorithm termed the generalized intersection method, which effectively searches a UAV path that meets the connectivity condition via transforming Problem 1-1 into a corresponding problem of discovering an shortest path over a weighted graph. We further show that the generalized intersection method guarantees an optimal UAV path in polynomial time. The generalized intersection method's pesudo code is delineated in Algorithm \ref{Algo1}. 
\begin{algorithm}
\caption{Generalized Intersection Method} \label{Algo1}
\textbf{Input:} $\mathbf{u}_0$, $\mathbf{u}_F$, $v_q$, $\mathbf{a}_m$, $d_0$, $\lambda_m$ for $m\in[1:M]$
\begin{algorithmic}[1]
\State \textbf{Def:} Function \textbf{ChkFea}($\mathbf{u}_0,\mathbf{u}_F,\mathbf{a}_m,d_0,\lambda_m$ for $m\in[1:M]$)  outputs a binary indicator regarding that Problem 1-1 is solvable $(h_\mathrm{fea}=1)$ or unsolvable $(h_\mathrm{fea}=0)$, where $\mathbf{u}_0$, $\mathbf{u}_F$, $\mathbf{a}_m$, $d_0$, and $\lambda_m$ for $m\in[1:M]$ described in Section \ref{sec2} are the parameters for the delivery environment.
\State \textbf{Def:} Function \textbf{Dijkstra}$(\mathbf{x}_1,\mathbf{x}_2,G)$ for graph $G=(V,E)$ outputs $(T,\mathbf{S}_V)$, where $\mathbf{x}_1$ and $\mathbf{x}_2$ are the vertices in $V$, $T$ denotes the minimum aggregate weight from $\mathbf{x}_1$ to $\mathbf{x}_2$ over the graph $G$, and $\mathbf{S}_V$ represents the corresponding optimal visiting sequence for vertices.
\State $V_0\leftarrow\{\mathbf{u}_0,\mathbf{u}_F\}$ \hfill\Comment{Initial and final points}
\State $E_0\leftarrow \emptyset$
\LeftComment{Step 1. Feasibility Check: Verify whether there exists a path from $\mathbf{u}_0$ to $\mathbf{u}_F$ under connectivity constraint.}
\State $h_\mathrm{fea} \leftarrow$ \textbf{ChkFea}$(\mathbf{u}_0,\mathbf{u}_F,\mathbf{a}_m,d_0,\lambda_m$ for $m\in[1:M])$
\If {$h_\mathrm{fea}=1$} \hfill\Comment{Problem 1-1 has a solution.}
    \LeftComment{Step 2. Vertex construction: Construct a set $V_0$ of vertices comprising the initial, final, and intersection points.}
    \For{$m,m'\in[1:M]$, $m<m'$} 
        \If{$\|\mathbf{a}_{m}-\mathbf{a}_{m'}\|\leq (d_0-\lambda_{m})+(d_0-\lambda_{m'})$}
            \State $V_0\leftarrow V_0\cup\{\mathbf{x}\in\mathbb{R}^2 \vert \ \|\mathbf{x}-\mathbf{a}_{m}\|=d_0-\lambda_{m},$
            \Statex \qquad\qquad \ $\|\mathbf{x}-\mathbf{a}_{m'}\|=d_0-\lambda_{m'}\}$
        \EndIf
    \EndFor 
    \LeftComment{Step 3. Edge construction: Construct a set $E_0$ of edges comprising the line segments falling within the total coverage map.}
    \For{$\mathbf{x}_1,\mathbf{x}_2\in V_0$, $\mathbf{x}_1\neq \mathbf{x}_2$}
        \State $h_\mathrm{out}\leftarrow\textbf{ChkOut}(\mathbf{x}_1,\mathbf{x}_2, \mathbf{a}_m, d_0,  \lambda_m$ for $m\in$
        \Statex \qquad\quad $[1:M])$
        \If {$h_\mathrm{out}=0$}
            \State $E_0\leftarrow E_0\cup (\mathbf{x}_1,\mathbf{x}_2,\|\mathbf{x}_1-\mathbf{x}_2\|/v_q)$
        \EndIf
    \EndFor 
    \LeftComment{Step 4. Path search: Search an optimal path from the initial point to the final point over the graph $G_0$.}
    \State $G_0\leftarrow(V_0,E_0)$ \hfill\Comment{Construct graph $G_0$}
    \State $(T,\mathbf{S}_{V_0})\leftarrow$ \textbf{Dijkstra}$(\mathbf{u}_0,\mathbf{u}_F,G_0)$
    \State $\mathbf{u}(t) \text{ for }t\in[0,T]\leftarrow$ \textbf{FindPath}$(\mathbf{S}_{V_0},v_q)$
\Else \Comment{Problem 1-1 does not have a solution.}
\State $T\leftarrow\infty$, $\mathbf{u}(t)\leftarrow \mathrm{Empty}$ for $t\in[0,T]$
\EndIf
\end{algorithmic}
\textbf{Output:} \big($h_\mathrm{fea}$, $T$, $\mathbf{u}(t)$ for $t\in[0,T]$\big)
\end{algorithm}
The algorithm initially verifies (in line $5$) whether there exists a path from the initial point $\mathbf{u}_0$ to the final point $\mathbf{u}_F$ (i.e., the problem is solvable) or not via the checking feasibility function ChkFea, that outputs a binary indicator regarding that Problem 1-1 is solvable $(h_\mathrm{fea}=1)$ or unsolvable $(h_\mathrm{fea}=0)$ under the inputs consisting of the initial and the final points, and the communication environment including the locations of the BSs and their effective coverage regions. We note that the function ChkFea can be established by leveraging \cite[Proposition~1]{Zhang:2019}, particularly when accommodating varying coverage radii among the BSs. For brevity, its pseudo code is excluded. If $h_\mathrm{fea}=1$ (i.e., the problem is solvable), we proceed to set up an undirected weighted graph $G_0=(V_0,E_0)$ by utilizing the intersection points of the effective coverage boundaries. In particular, the vertex set $V_0$ is composed of the initial point $\mathbf{u}_0$, the final point $\mathbf{u}_F$, and all intersection points (in lines $7$-$11$).
The edge set involves a line segment $\overline{\mathbf{x}_1\mathbf{x}_2}$ connecting two distinct vertices $\mathbf{x}_1,\mathbf{x}_2\in V_0$ that falls within the total coverage map, which refers to the set of all the effective coverage regions of BSs (in lines $12$-$17$). We note that whether a line segment falls within the total coverage map can be validated via the checking outage function ChkOut, where its pesudo code is delineated in Algorithm \ref{Algo3} and elaborated upon subsequently. This edge is represented as a tuple $(\mathbf{x}_1,\mathbf{x}_2,\|\mathbf{x}_1-\mathbf{x}_2\|/{v_q})$, where the edge weight $\|\mathbf{x}_1-\mathbf{x}_2\|/{v_q}$ signifies the travel time over the line segment $\overline{\mathbf{x}_1\mathbf{x}_2}$. Following the construction of the graph $G_0=(V_0,E_0)$, we then derive an optimal UAV trajectory from $\mathbf{u}_0$ to $\mathbf{u}_F$ and the corresponding optimal travel time $T$ over the graph (in lines $19$-$20$). We initially determine an optimal visiting sequence $\mathbf{S}_{V_0}$ for vertices over the graph and its aggregate weight (i.e., travel time $T$) through the Dijkstra algorithm \cite{Dijkstra:1959}, which outputs a path with minimum aggregate weight from a vertex to another vertex over a weighted graph with polynomial complexity. After that, the final UAV trajectory is determined using the finding path function FindPath, that outputs the trajectory associated with the optimal visiting sequence $\mathbf{S}_{V_0}$ and the maximum flying speed $v_q$. The function FindPath can be implemented in a similar manner as outlined in \cite[$(25)$-$(27)$]{Zhang:2019}, where the pseudo code is excluded for brevity. An example of the graph $G_0$ and the corresponding optimal trajectory by Algorithm \ref{Algo1} is illustrated in Fig. \ref{Fig3}.


\begin{figure}[t]
\centering
\includegraphics[width=0.9\columnwidth]{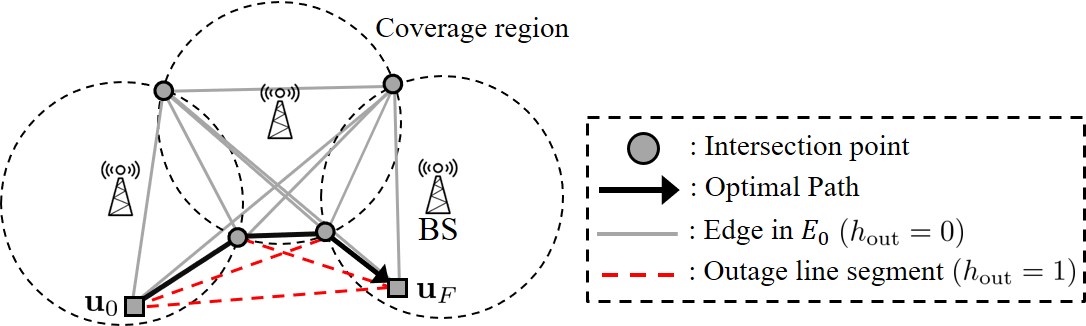}
\caption{An example of graph $G_0$ for $M=3$. The graph has vertex set $V_0$ including $\mathbf{u}_0$, $\mathbf{u}_F$, and all intersection points and edge set $E_0$ including the (solid) line segments between two vertices which fall within the total coverage map.}\label{Fig3}
\end{figure}

Algorithm \ref{Algo3} outlines the function ChkOut that evaluates whether each line segment $\overline{\mathbf{x}_1\mathbf{x}_2}$ connecting two distinct vertices $\mathbf{x}_1,\mathbf{x}_2\in V_0$ falls within the total coverage map. We define that the line segment experiences an outage if at least one of $\xi\in[0,1]$ meets the following inequality:
\begin{align}
\min_{m\in\mathcal[1:M]}\|\pmb{\alpha}(\xi)-\mathbf{a}_m\|+\lambda_m>d_0,\label{eq:25}
\end{align}
where $\pmb{\alpha}(\xi)\triangleq \mathbf{x}_1+\xi(\mathbf{x}_2-\mathbf{x}_1)$ for $\xi\in[0,1]$ denotes a point along the line segment $\overline{\mathbf{x}_1\mathbf{x}_2}$. Here, \eqref{eq:25} signifies that the point $\pmb{\alpha}(\xi)$ does not fall within the total coverage map, i.e., the UAV can not establish a connection with any BS. To test whether the line segment experiences an outage, the function ChkOut confirms the presence of a $\xi\in[0,1]$ that fulfills \eqref{eq:25}. We first define the safe interval $\mathcal{T}_\mathrm{safe}\triangleq [0,\xi']$ for a $\xi'\in[0,1]$ as the line segment between $\mathbf{x}_1$ and $\pmb{\alpha}(\xi')$ where every $\pmb{\alpha}(\xi)$ for $\xi\in \mathcal{T}_\mathrm{safe}$ has been verified to fall within the total coverage map. In other words, any $\xi\in \mathcal{T}_\mathrm{safe}$ satisfies $\|\pmb{\alpha}(\xi)-\mathbf{a}_m\|+\lambda_m \leq d_0$ at least for an $m\in[1:M]$. This function initially verifies the presence of $\xi=0$ within the total coverage map. Subsequently, it iteratively updates the safe interval or proclaims an outage as the following procedures. Let us assume that the existing safe interval is provided as $[0,\xi']\subseteq [0,1]$. If the point $\pmb{\alpha}(\xi'+\epsilon)$ falls within the effective coverage region of $\mathrm{BS}_m$, where $\epsilon>0$ is a sufficiently small positive constant, then we extend the safety interval through encompassing every $\xi\in[0,1]$ wherein $\pmb{\alpha}(\xi)$ falls within the effective coverage region of $\mathrm{BS}_m$, i.e.,  $\|\pmb{\alpha}(\xi)-\mathbf{a}_m\|\leq d_0-\lambda_m$. 
The function ChkOut terminates either when $\pmb{\alpha}(\xi'+\epsilon)$ does not fall within the total coverage map  $(h_\mathrm{out}=1)$ or the safe interval $\mathcal{T}_\mathrm{safe}$ eventually spans the entire range of $[0,1]$ $(h_\mathrm{out}=0)$. It outputs a binary indicator regarding that the line segment $\overline{\mathbf{x}_1\mathbf{x}_2}$ falls within the total coverage map $(h_\mathrm{out}=0)$ or experiences an outage $(h_\mathrm{out}=1)$.\footnote{We note that the function ChkOut allows an outage over a path with length up to $\epsilon\|\mathbf{x}_2-\mathbf{x}_1\|$. To make such an outage  probability negligible, it is needed to select a sufficiently small $\epsilon>0$.} 
An example of updating the safe interval is shown in Fig. \ref{Fig4}.

\begin{remark}\label{rmk2}
The idea of repeatedly updating the safe interval is similar to the short-cut edge concept in \cite[Algorithm~2]{Chapnevis:2021}, which allows  communication outage duration up to a certain threshold for edge construction. Since our algorithm does not permit for an edge to experience an outage, the maximum number of updates to the safe interval is reduced from $2M$ to $M$ compared to the algorithm in \cite{Chapnevis:2021}.
\end{remark}

\begin{algorithm}[t]
\caption{Function ChkOut} \label{Algo3}
\textbf{Input:} $\mathbf{x}_1,\mathbf{x}_2\in V_0$, $\mathbf{a}_m$, $d_0$, $\lambda_m$ for $m\in[1:M]$
\begin{algorithmic}[1]
\State\textbf{Def:} $\pmb{\alpha}(\xi)\triangleq \mathbf{x}_1+\xi(\mathbf{x}_2-\mathbf{x}_1)$ for $\xi\in[0,1]$
\LeftComment{$\epsilon$ is a sufficiently small positive constant.}
\State $\xi'\leftarrow 0$, $\xi''\leftarrow 0$, $h_\mathrm{out}\leftarrow 0$, $\epsilon\leftarrow 10^{-10}$
\While{$\xi'<1$}   
\LeftComment{Update safe interval $\mathcal{T}_\mathrm{safe}$ from 
$[0,\xi']$ to $[0,\xi'']$ if $\pmb{\alpha}(\xi'+\epsilon)$ falls within the total coverage map.}
    \For {$m\in\mathcal[1:M]$}  \hfill\Comment{Find BS covering $\pmb{\alpha}(\xi'+\epsilon)$.}
        \If{$\|\pmb{\alpha}(\xi'+\epsilon)-\mathbf{a}_m\|\leq d_0-\lambda_m$}
            \State $\xi''\!\leftarrow\!\max\{\xi\in [0,1] \vert \ \|\pmb{\alpha}(\xi)-\mathbf{a}_m\|\leq d_0-\lambda_m\}$
            \State \textbf{break} 
        \EndIf
    \EndFor
    \If{$\xi''=\xi'$}  \hfill\Comment{ $\pmb{\alpha}(\xi'+\epsilon)$ experiences an outage.}
         \State $h_\mathrm{out}\leftarrow 1$
         \State  \textbf{break}
    \EndIf
    \State $\xi'\leftarrow\xi''$
\EndWhile \hfill\Comment{$\xi'=1$ indicates that $\mathcal{T}_\mathrm{safe}=[0,1]$.}
\end{algorithmic}
\textbf{Output:} $h_\mathrm{out}$
\end{algorithm}

\begin{figure}
\centering
\includegraphics[width=0.8\columnwidth]{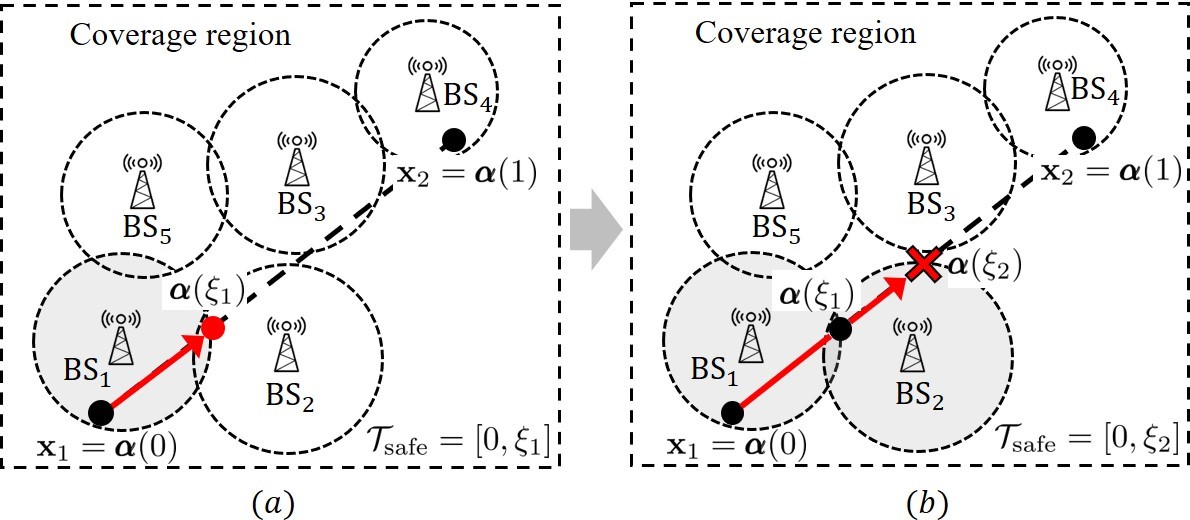}
\caption{An example of updating the safe interval $\mathcal{T}_\mathrm{safe}$ in line segment $\overline{\mathbf{x}_1\mathbf{x}_2}$ for $M=5$. $(a)$ It first checks whether $\pmb{\alpha}(0)$ is connected with $\mathrm{BS}_1$ and then refines the safe interval $\mathcal{T}_\mathrm{safe}=[0,\xi_1]$ by considering the coverage region of $\mathrm{BS}_1$. $(b)$ Since the UAV at $\pmb{\alpha}(\xi_1+\epsilon)$ is connected with $\mathrm{BS}_2$, the safe interval is updated to $\mathcal{T}_\mathrm{safe}=[0,\xi_2]$. Next, the UAV at $\pmb{\alpha}(\xi_2+\epsilon)$ is not connected with every BS and hence $h_\mathrm{out}=1$.}\label{Fig4}
\end{figure}

Now, the following theorems show that our generalized intersection method yields an optimal solution of Problem 1-1 in polynomial time.

\begin{theorem}\label{Thm1}
The generalized intersection method outputs an optimal solution for Problem 1-1.
\end{theorem}
\begin{proof}
It was previously shown in \cite[Proposition~3]{Zhang:2019} that an optimal solution of Problem 1-1 consists of line segments, where its breakpoints are selected in the overlapping regions of the coverage regions of two different BSs. However, under the approach in \cite{Zhang:2019}, we can not derive an optimal solution of the problem via a graph-theoretic approach, since there are infinite number of possible points which can serve as breakpoints in each overlapping region, and the number of vertices in a graph should be finite to apply the Dijkstra algorithm \cite{Dijkstra:1959}. Following the result of \cite{Zhang:2019}, in this proof, we show that the breakpoints of an optimal path should be selected in the intersection points of the coverage boundaries of BSs in order to justify adopting a graph theory-based approach in our algorithm. Note that the problem is equivalent to deriving a path which achieves the shortest distance under the connectivity constraint since the speed of the UAV is fixed at $v_q$.

For a proof by contradiction, let us assume that an optimal path of the UAV has a breakpoint $\mathbf{x}_\mathrm{br}$ in the overlapping region of $\mathrm{BS}_1$ and $\mathrm{BS}_2$ except the corresponding intersection points. Then, there exists sufficiently small  $\delta>0$ that the set  $\mathcal{R}_\delta\triangleq\{\mathbf{x}\in\mathbb{R}^2| \ \|\mathbf{x}-\mathbf{x}_\mathrm{br}\|\leq \delta\}$ is included in the total coverage map because $\|\mathbf{x}_\mathrm{br}-\mathbf{a}_m\|<d_0-\lambda_m$ at $m=1$ or $2$, 
which means that the point $\mathbf{x}_\mathrm{br}$ is inside the coverage region of $\mathrm{BS}_1$ or $\mathrm{BS}_2$ except its coverage boundary. Now, let us denote $\pmb{\beta}_1$ and $\pmb{\beta}_2$ as two intersections of the boundary of $\mathcal{R}_\delta$ and the path of the UAV. Then, $\|\pmb{\beta}_1-\pmb{\beta}_2\|< \|\pmb{\beta}_1-\mathbf{x}_\mathrm{br}\|+\|\mathbf{x}_\mathrm{br}-\pmb{\beta}_2\|$ by triangular inequality, where we note that only strict inequality holds since the point $\mathbf{x}_\mathrm{br}$ is a breakpoint of the path of the UAV.
The path of the UAV includes the line segments $\overline{\pmb{\beta}_1\mathbf{x}_\mathrm{br}}$ and $\overline{\mathbf{x}_\mathrm{br}\pmb{\beta}_2}$. Hence, it is a contradiction that the path is an optimal solution for Problem 1-1 because the overall length of the path can be strictly decreased by substituting $\overline{\pmb{\beta}_1\pmb{\beta}_2}$ for $\overline{\pmb{\beta}_1\mathbf{x}_\mathrm{br}}$ and $\overline{\mathbf{x}_\mathrm{br}\pmb{\beta}_2}$ as shown in Fig. \ref{Fig6}.
\end{proof}
\begin{figure}
\centering
\includegraphics[width=0.8\columnwidth]{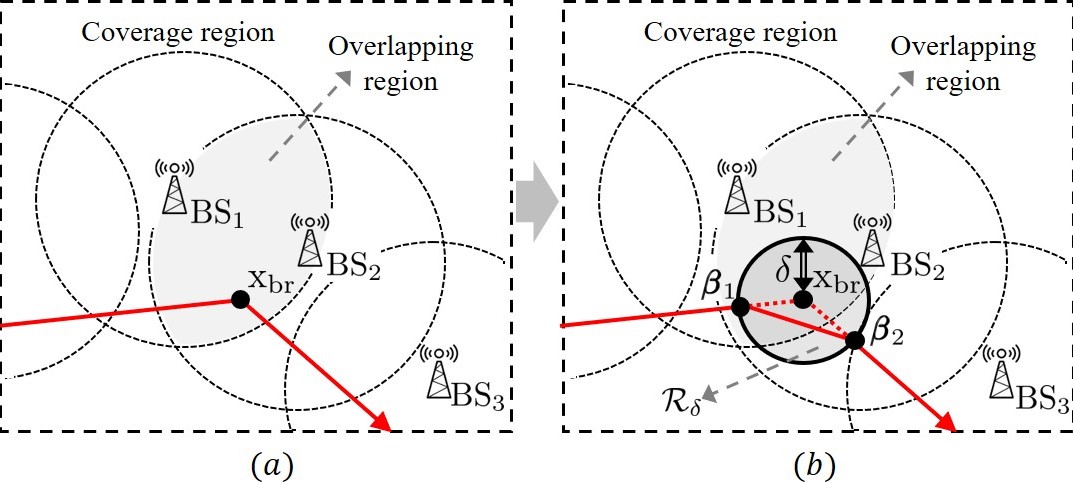}
\caption{An example of proof of Theorem \ref{Thm1}. In Fig. \ref{Fig6}-$(a)$, the breakpoint $\mathbf{x}_\mathrm{br}$ of the path is not an intersection point. This path is not an optimal solution of Problem 1-1 since there exists a shorter path compared to the path in \ref{Fig6}-$(a)$ as shown in \ref{Fig6}-$(b)$.}\label{Fig6}
\end{figure}

\begin{table*}
\caption{Comparison of algorithms for Problem 1-1}\label{Tab1}
\centering
\begin{tabular}{@{} c || c | c @{}}
\cline{1-3}
Algorithm & Complexity & Performance gap\\ \cline{1-3}
Exhaustive search \cite{Zhang:2019} & $O(M!M^{3.5})$ & 0\\ \cline{1-3}
Exhaustive search with fixed association \cite{Zhang:2019} & $O(M^{3.5})$ & $O(Md_0/{v_q})$ \\ \cline{1-3}
Exhaustive search with quantization \cite{Zhang:2019} & $O(M^4Q^2)$ & $O((Md_0/{v_q})\sin(1/{Q}))$ \\ \cline{1-3}
Intersection method \cite{Chen:2020} by checking outages via Algorithm \ref{Algo3} & $O(M^4)$ & $O(Md_0/{v_q})$ \\ \cline{1-3}
Ours (Generalized intersection method)  & $O(M^6)$ & $0$ \\ \cline{1-3}
\end{tabular}
\end{table*}

\begin{theorem}\label{Thm2}
The time complexity of the generalized intersection method is $O(M^6)$.
\end{theorem}
\begin{proof}
Let us first state the cardinality of the set $|V_0|=O(M^2)$. The steps in Algorithm \ref{Algo1} have the following complexities:
\begin{itemize}
\item Complexity of function ChkFea: It was shown that the complexity to check whether Problem 1-1 is feasible is $O(M^2)$ \cite{Zhang:2019}.
\item Step 1. Vertex construction: This step has complexity $O(M^2)$ since the intersection points of the coverage boundaries by a BS pair is derived by calculating the quadratic equations in line $9$ of Algorithm \ref{Algo1} and the number of the possible BS pairs is $O(M^2)$.
\item Step 2. Edge construction: First, we derive the complexity of testing whether a line segment experiences an outage via Algorithm \ref{Algo3}. The number of updates to the safe interval (i.e., the number of iterations of the \textbf{while} loop) is at most $M$ since in each update, the BSs used in previous updates are not selected again. For each update, the complexity order to derive the next $\xi'$ (i.e., $\xi''$) is $O(M)$ since the \textbf{for} loop in lines $4$-$9$ is repeated at most $M$ times. Hence, the complexity of Algorithm \ref{Algo3} is $M\cdot O(M)=O(M^2)$. Next, we find the number of the line segments that need to be tested. Since the number of the vertices in the graph $G_0$ is $|V_0|=O(M^2)$, the number of the line segments is $|V_0|^2=O(M^4)$. Consequently, the complexity of testing all line segments to construct the edge set $E_0$ is  $O(M^2)\cdot O(M^4)=O(M^6)$. 
\item Step 3. Path search: The complexity of the Dijkstra algorithm in the graph $G_0$ is $O(|V_0|^2)=O(M^4)$ \cite{West:2001}.
\end{itemize}
Consequently, the complexity of the generalized intersection method is $O(M^6)$, which is dominated at the edge $E_0$ construction step.
\end{proof}

Now, we turn our attention to compare our generalized intersection method with several existing algorithms in \cite{Zhang:2019,Chen:2020} addressing Problem 1-1. Table \ref{Tab1} offers an overview of the time complexity order and the gap in performance from the optimal mission time across the algorithms. We present succinct explanations of past algorithms and insights derived from Table \ref{Tab1}.
\begin{itemize}
    \item Out of the algorithms listed in Table \ref{Tab1}, only our approach yields an optimal solution in polynomial time.
    \item The three existing algorithms of exhaustive search (ES), exhaustive search with fixed association (ES-FA), and exhaustive search with quantization (ES-Q) are introduced in \cite{Zhang:2019}. It is proved in \cite[Proposition~3]{Zhang:2019} that an optimal UAV path for Problem 1-1 is composed of finite line segments and the corresponding breakpoints which fall within the overlapping regions of coverage regions of two distinct BSs. Under such methodology, finding an optimal UAV path by utilizing the Dijkstra algorithm is impossible because there are infinite number of points in each overlapping region. The ES algorithm \cite{Zhang:2019} is designed to determine optimal breakpoints within overlapping regions via convex optimization-based algorithms, which is an optimal algorithm with non-polynomial time complexity over $M$. To alleviate the computing time, \cite{Zhang:2019} also presents two sub-optimal algorithms, namely ES-FA and ES-Q algorithms which have polynomial time complexity. The ES-FA algorithm closely resembles the ES algorithm in its core approaches except that it predetermines the BS association sequence as a fixed one. The ES-Q algorithm adopts a graph theory-based approach, converting each overlapping region into a discrete set of points through quantization. The ES-FA algorithm achieves the time complexity lower than the generalized intersection method, yet its performance gap widens with the increase of $M$. Regarding the ES-Q algorithm, we denote $Q\in\mathbb{N}$ as the number of the quantization points within individual overlapping regions. For this algorithm, the performance gap widens with the increase of $M$ when $Q=O(M)$, and its complexity exceeds that of our generalized intersection method when $Q=\omega(M)$.
    \item The intersection method introduced in \cite{Chen:2020} just uses the finite intersection points as potential breakpoints and adopts a graph theory-based approach similar to ours. Note that the algorithm is not optimal due to its reliance on a predetermined BS association sequence, heuristically set as in the ES-FA algorithm \cite{Zhang:2019}. Moreover, it does not clearly propose a function that evaluates whether each line segment connecting two distinct vertices in the graph falls within the total coverage map, like the ChkOut function in Algorithm \ref{Algo3}. Applying the ChkOut function into the intersection method \cite{Chen:2020}, we reveal that the algorithm achieves the same performance gap but has a higher complexity, compared to the ES-FA algorithm \cite{Zhang:2019}.
\end{itemize}

\section{Optimal Trajectory with the Connectivity and Battery Constraints}\label{sec4}
The section aims to plan a UAV path with minimum mission time under connectivity and battery constraints, i.e., addresses Problem 1. We point out that under limited battery capacity, adjusting the flying speed $v$ while traveling can yield advantages because the maximum flight distance in \eqref{eq:3} without battery replacement varies over $v$. For ease of notation, $H_\mathrm{CS}=H$ is assumed. 

To solve Problem 1, we introduce a modified version of our generalized intersection method in Section \ref{sec3}, termed as the generalized intersection method with battery constraint (GIM-B). In addition, we verify that the proposed GIM-B algorithm guarantees an optimal UAV trajectory in polynomial time. This algorithm's pesudo code is delineated in Algorithm \ref{Algo4}.
\begin{algorithm}
\caption{Generalized Intersection Method with Battery Constraint (GIM-B)} \label{Algo4}
\textbf{Input:} $\mathbf{u}_0$, $\mathbf{u}_F$, $\mathcal{V}$, $\mathbf{a}_m$, $d_0$, $\lambda_m$, $\mathbf{c}_n$, $\tau_{C_n}$, $w$, $w_2$,  for $m\in\mathcal[1:M]$, $n\in\mathcal[1:N]$
\begin{algorithmic}[1]
\State \textbf{Def:} Function \textbf{BFS}$(\mathbf{x}_1,\mathbf{x}_2,G)$ for graph $G=(V,E)$ outputs a binary indicator regarding that the vertices $\mathbf{x}_1,\mathbf{x}_2\in V$ are connected in the graph $G$ $(h_\mathrm{Gfea}=1)$ or not $(h_\mathrm{Gfea}=0)$.
\State $V_\mathrm{GL}\leftarrow\{\mathbf{u}_0,\mathbf{u}_F,\mathbf{c}_1,...,\mathbf{c}_N\}$, $V_\mathrm{LO},E_\mathrm{LO},E_\mathrm{GL}\leftarrow \emptyset$,
\Statex $V_\mathrm{in},E_\mathrm{in},E_1,...,E_{N+2} \leftarrow \emptyset$
\LeftComment{Treat the initial and the final points as CSs.}
\State $\mathbf{c}_{N+1} \leftarrow \mathbf{u}_0$, $\mathbf{c}_{N+2} \leftarrow \mathbf{u}_F$, $\tau_{C_{N+1}},\tau_{C_{N+2}}\leftarrow 0$
\State $V_\mathrm{in}\!\leftarrow\! \text{All intersection points}$ \hfill\Comment{Lines $7$-$11$ at Algorithm \ref{Algo1}}
\State $V_\mathrm{all}\leftarrow V_\mathrm{GL}\cup V_\mathrm{in}$
\LeftComment{Step 1. Outage test: Evaluate whether each conceivable line segment experiences an outage.}
\For{$\mathbf{x}_1,\mathbf{x}_2\in V_\mathrm{all}$, $\mathbf{x}_1\neq \mathbf{x}_2$} 
    \State $h_\mathrm{out}\leftarrow$ \textbf{ChkOut}$(\mathbf{x}_1,\mathbf{x}_2,\mathbf{a}_m, d_0, \lambda_m \!\text{ for }  m\in[1:M])$
    \For {$n\in [1:N+2]$}
        \If{$h_\mathrm{out}=0$, $\mathbf{c}_{n}\in \{\mathbf{x}_1,\mathbf{x}_2\}$ }
            \State $E_n\leftarrow E_n\cup (\mathbf{x}_1,\mathbf{x}_2,\|\mathbf{x}_1-\mathbf{x}_2\|)$
        \EndIf
    \EndFor
    \If{$h_\mathrm{out}\!=\!0$, $\mathbf{c}_{n}\!\not\in\!\{\mathbf{x}_1,\mathbf{x}_2\}$ for $n\in[1:N+2]$}
        \State $E_\mathrm{in}\leftarrow E_\mathrm{in}\cup (\mathbf{x}_1,\mathbf{x}_2,\|\mathbf{x}_1-\mathbf{x}_2\|)$
    \EndIf
\EndFor 
\LeftComment{Step 2. Local level search: Find a shortest path between each pair of CSs.}
\For{$n\in [1:N+1]$, $n'\in[1:N]\cup \{N+2\}$, $n\neq n'$} 
    \LeftComment{Function \textbf{ChkFea} is described in line $1$ at Algorithm \ref{Algo1}.} 
    \State $h_\mathrm{Lfea}\!\leftarrow\!$ \textbf{ChkFea}$(\mathbf{c}_n,\mathbf{c}_{n'}, \mathbf{a}_m, d_0, \lambda_m \!\text{ for } m\in[1:M])$ 
    \If{$h_\mathrm{Lfea}=1$}
        \State $V_\mathrm{LO}\leftarrow V_\mathrm{in}\cup \{\mathbf{c}_n, \mathbf{c}_{n'}\}$, $E_\mathrm{LO}\leftarrow E_\mathrm{in}\cup E_n\cup E_{n'}$
        \State $G_\mathrm{LO}\leftarrow (V_\mathrm{LO}, E_\mathrm{LO})$
        \LeftComment{Function \textbf{Dijkstra} is described in line $2$ at Algorithm \ref{Algo1}.}
        \State $(\ell_\mathrm{LO},\mathbf{S}_{V_\mathrm{LO}}(c_n, c_{n'}))\leftarrow \textbf{Dijkstra}(\mathbf{c}_n,\mathbf{c}_{n'},G_\mathrm{LO})$
        \State ($h_\mathrm{sp},v(c_n, c_{n'})) \leftarrow$ \textbf{ChkSp}$(\ell_\mathrm{LO},\mathcal{V},w,w_2)$
        \If{$h_\mathrm{sp}=1$}
            \State $E_\mathrm{GL}\leftarrow E_\mathrm{GL}\cup (\mathbf{c}_n,\mathbf{c}_{n'},\ell_\mathrm{LO}/v(c_n, c_{n'})+\tau_{C_{n'}})$
        \EndIf
    \EndIf
\EndFor  
\LeftComment{Step 3. Global level search: Find an optimal path from the initial point to the final point over the graph of CSs.}
\State $\overrightarrow{G}_\mathrm{GL}\leftarrow (V_\mathrm{GL}, E_\mathrm{GL})$ \hfill\Comment{$\overrightarrow{G}_\mathrm{GL}$ is a directed graph.}
\State $h_\mathrm{Gfea}\leftarrow$ \textbf{BFS}$(\mathbf{u}_0,\mathbf{u}_F,\overrightarrow{G}_\mathrm{GL})$
\If{$h_\mathrm{Gfea}=1$}
    \State ($T$, $\mathbf{S}_{V_\mathrm{GL}})\leftarrow$ $\textbf{Dijkstra}(\mathbf{u}_0,\mathbf{u}_F,\overrightarrow{G}_\mathrm{GL})$
    \State ($\mathbf{u}(t)$, $\psi(t)$ for $t\in[0,T]$) $\leftarrow$ \textbf{FindPathG}$(\mathbf{S}_{V_\mathrm{GL}},$
    \Statex\qquad $v(c_n, c_{n'})$, $\mathbf{S}_{V_\mathrm{LO}}(c_n, c_{n'})$, $\tau_{C_{n'}}$ for $n\in[1:N+1],$
    \Statex\qquad $n'\in [1:N]\cup \{N+2\}$)
\Else
    \State $h_\mathrm{Gfea}\leftarrow 0$, $T\leftarrow\infty$, $\mathbf{u}(t),\psi(t) \leftarrow \mathrm{Null}$ for $t\in[0,T]$
\EndIf   
\end{algorithmic}
\textbf{Output:} \big($h_\mathrm{Gfea}$, $T$, $\mathbf{u}(t)$, $\psi(t)$ for $t\in[0,T]$\big)
\end{algorithm}
The outputted trajectory between the initial and final points yielded by the GIM-B algorithm is composed of finite line segments, where its breakpoints are selected in the CSs and the intersection points of the coverage boundaries. In our algorithm, such line segments (i.e., edges in the equivalent graphs) which make up the UAV path are determined through a three-step process. \sh{In Step 1 (lines $6$-$16$ of Algorithm \ref{Algo4}), it evaluates whether each line segment falls within the total coverage map (i.e., not experiences an outage) and constructs the edge sets without an outage. Next, our algorithm determines an optimal UAV trajectory through two-level process. In Step 2 corresponding to local level path finding (lines $17$-$28$), by regarding the initial and final points also as CSs, it first 
determines a shortest UAV path between each pair of CSs via Dijkstra algorithm on  the graph whose vertex set consists of the corresponding pair of CSs and all intersection points and edge set consists of the line segments between two vertices that survived in Step 1, as illustrated in Fig. \ref{Fig7}. We note that the line segments experiencing an outage are excluded in Step 1 and thus not included in the edge set for finding a shortest path between each pair of CSs in Step 2.} Then, it determines the corresponding maximum permissible speed under the battery limit via the checking speed function ChkSp whose pseudo code is described in Algorithm \ref{Algo5}.
In the local level, it is assumed that the UAV flies with a fixed speed while traveling between each pair of CSs, which is justified later in Theorem~\ref{Thm3}. 
We point out that traveling between two distinct CSs may be impossible since there may be no feasible path between them  ($h_\mathrm{Lfea}=0$) or the path's length may exceed the maximum flight distance for any possible speed under the battery limit ($h_\mathrm{sp}=0$). 
\begin{figure}
\centering
\includegraphics[width=0.8\columnwidth]{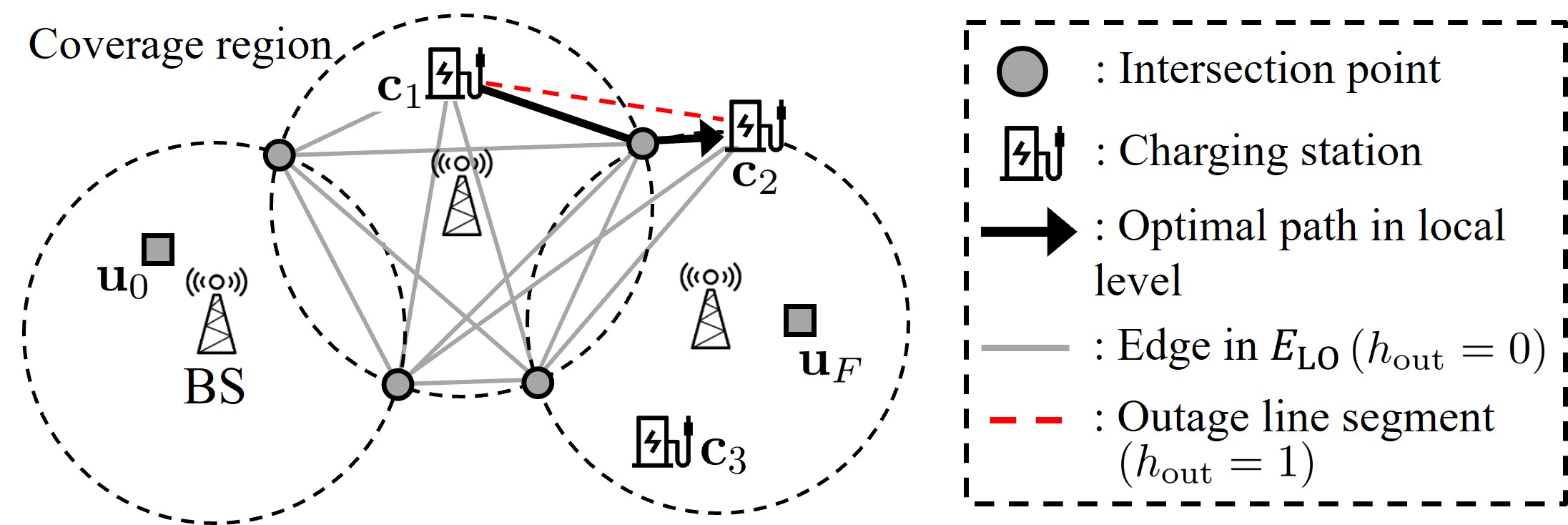}
\caption{ \sh{An example of graph $G_\mathrm{LO}=(V_\mathrm{LO},E_\mathrm{LO})$ for finding the shortest path from $\mathbf{c}_1$ to $\mathbf{c}_2$ in Step 2 of Algorithm \ref{Algo4} with $M=3$ and $N=3$. The vertex set $V_\mathrm{LO}$ consists of $\mathbf{c}_1$, $\mathbf{c}_2$, and all intersection points and the edge set $E_\mathrm{LO}$ includes solid line segments between two vertices which fall within the total coverage map.}}\label{Fig7}
\end{figure}
Finally, in Step 3 corresponding to global level path finding (lines $29$-$36$), our algorithm constructs a directed graph consisting of the CSs as vertices and the pairs of the CSs that have been confirmed to be accessible in the local level as edges, where the weight of each edge is given as the sum of the airborne flight time over the edge and the delay for battery replacement at the arrived CS as specified in line $25$.
To construct the graph, it initially 
examines the possibility of traveling between the initial and the final points by the breadth-first search function BFS \cite{West:2001}, that explores every vertex connected with a start vertex in a graph in polynomial time. If possible, the algorithm derives an optimal visiting sequence for CSs via the Dijkstra algorithm over the graph, followed by utilizing the finding path (in global level) function FindPathG that yields an optimal UAV trajectory in response to the visiting sequences for vertices in the local and global levels, the flight speeds between CS pairs, and the delays for battery replacement. An example of the graph to derive an optimal path from the initial point to the final point at the global level is illustrated in Fig. \ref{Fig8}.
\begin{figure}
\centering
\includegraphics[width=0.8\columnwidth]{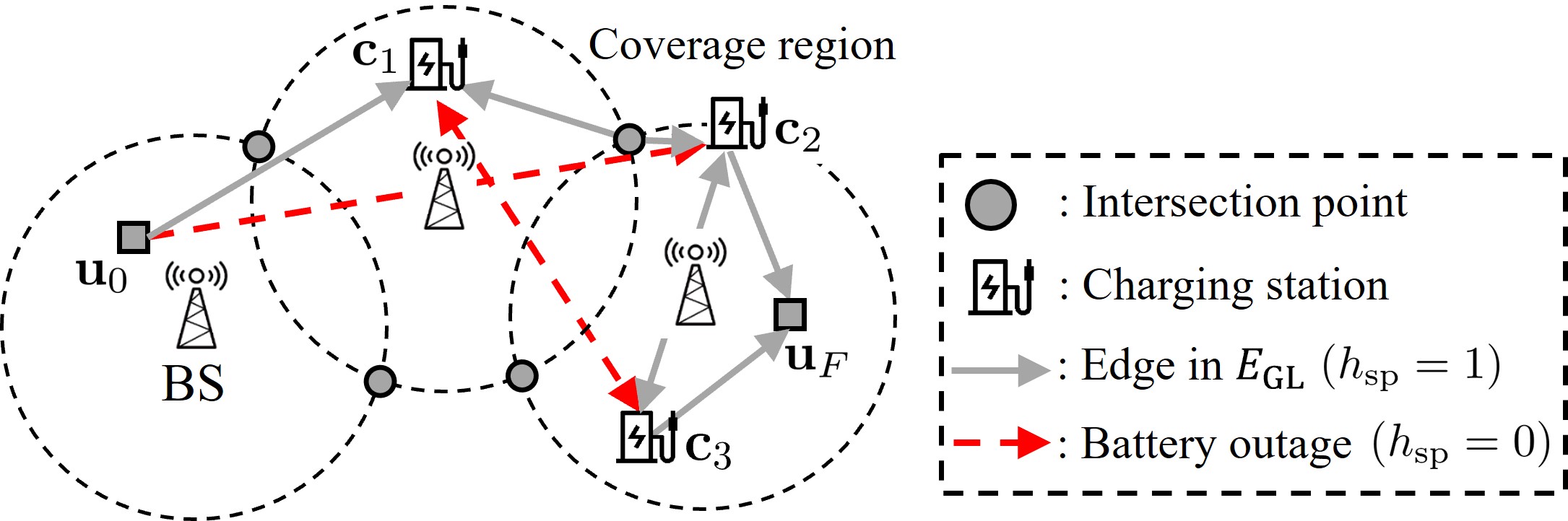}
\caption{An example of graph $\protect\overrightarrow{G}_\mathrm{GL}=(V_\mathrm{GL},E_\mathrm{GL})$ in the global level of Algorithm \ref{Algo4}, where $M=3$ and $N=3$.}\label{Fig8}
\end{figure}

Algorithm \ref{Algo5} describes the function ChkSp which checks whether the UAV can travel a distance $\ell_\mathrm{LO}\geq 0$ without battery replacement $(h_\mathrm{sp}=1)$ or not $(h_\mathrm{sp}=0)$ by using the maximum possible traveling distance function $d_\mathrm{fly}(v)$ in \eqref{eq:3} for speed $v\in\mathcal{V}$. If it is possible $(h_\mathrm{sp}=1)$, then it derives the maximum possible speed $v_\mathrm{max}\in\mathcal{V}$ whose maximum traveling distance $d_\mathrm{fly}(v_\mathrm{max})$ is not smaller than $\ell_\mathrm{LO}$. We note that the algorithm assumes that the UAV maintains a fixed speed while traveling between two CSs, but the speed can vary depending on the pair of CSs. The following theorem shows a sufficient condition for traveling between two CSs with a fixed speed to be optimal.

\begin{algorithm}[t]
\caption{Function ChkSp} \label{Algo5}
\textbf{Input:} $\ell_\mathrm{LO}$, $\mathcal{V}$, $w$, $w_2$
\begin{algorithmic}[1]
\If{$\{v\in\mathcal{V}|d_\mathrm{fly}(v)\geq \ell_\mathrm{LO}\}\neq \emptyset$}
    \State $h_\mathrm{sp}\leftarrow 1$ \hfill\Comment{Can travel $\ell_\mathrm{LO}$ without battery replacement} 
    
    \LeftComment{Find the maximum possible speed $v_\mathrm{max}$ that can travel the length $\ell_\mathrm{LO}$ without battery replacement.}
    \State $v_\mathrm{max}\leftarrow \max_{v\in\mathcal{V}}\{v|d_\mathrm{fly}(v)\geq \ell_\mathrm{LO}\}$
\Else
    \State $h_\mathrm{sp}\leftarrow 0$, $v_\mathrm{max}=0$
\EndIf
\end{algorithmic}
\textbf{Output:} $(h_\mathrm{sp}$, $v_\mathrm{max})$
\end{algorithm}

\begin{theorem}\label{Thm3} 
Assume that the UAV can fly with any speed $v\in[v_1,v_q]$ and the power consumption model $P_\mathrm{UAV}(v)$ is convex for $v\in[v_1,v_q]$. Then, for traveling between two CSs with the connectivity and battery constraints, flying with a fixed speed minimizes the traveling time.
\end{theorem}
\begin{proof}
Let us assume that the path distance $\ell_\mathrm{LO}$ to travel between two CSs is partitioned by segments $\ell_1,...,\ell_K$ where $\ell_\mathrm{LO}=\sum_{k=1}^K\ell_k$ and the UAV flies with speed $\tilde{v}_k\in[v_1,v_q]$ for segment $\ell_k$ for $k\in[1:K]$. In this case, we have  the total travel time $T_\mathrm{LO}=\sum_{k=1}^K \ell_k/\tilde{v}_k$ and the total consumed energy $E_\mathrm{LO}=\sum_{k=1}^K (\ell_k/\tilde{v}_k)\cdot P_\mathrm{UAV}(\tilde{v}_k)$. We prove this theorem by showing the UAV can travel $\ell_\mathrm{LO}$ within time $T_\mathrm{LO}$ by a fixed speed $\bar{v}\in[v_1,v_q]$ while consuming energy equal to or less than $E_\mathrm{LO}$. First, the UAV can travel $\ell_\mathrm{LO}$ in time $T_\mathrm{LO}$ if it travels with the fixed speed $\bar{v}={\ell_\mathrm{LO}\over{\sum_{k'=1}^K \ell_{k'}/\tilde{v}_{k'}}}$. Second, $E_\mathrm{LO}$ is lower-bounded as:
\begin{align}
E_\mathrm{LO}&=\sum_{k=1}^K (\ell_k/\tilde{v}_k)\cdot P_\mathrm{UAV}(\tilde{v}_k) \label{eq:30}\\
\overset{(a)}\geq& \Biggl(\sum_{k'=1}^K \ell_{k'}/\tilde{v}_{k'}\!\Biggr)\!\cdot P_\mathrm{UAV}\left(\sum_{k=1}^K{{\ell_k/\tilde{v}_k}\over{\sum_{k'=1}^K \ell_{k'}/\tilde{v}_{k'}}}\cdot \tilde{v}_k\!\!\right) \label{eq:31}\\
=& T_\mathrm{LO} \cdot P_\mathrm{UAV}(\bar{v}), \label{eq:32}
\end{align}
where $(a)$ is by Jensen's inequality. Since the UAV with fixed speed $\bar{v}$ consumes less energy than $E_\mathrm{LO}$ as $\eqref{eq:32}$, this proves the theorem.  
\end{proof}
We note that the power consumption model in \eqref{eq:1} can be approximated as a convex function when $v\gg v_0(w)$ as proved in \cite{Zeng:2019}. 
Hence, in Algorithm \ref{Algo1}, traveling between two CSs with a fixed speed, while the speed can vary depending on the pair of CSs, is close to optimal.

Now, the following theorems show that our GIM-B algorithm outputs an optimal solution of Problem 1 in polynomial time under the assumption that the power consumption model $P_\mathrm{UAV}(v)$ is convex in the range of the UAV  speed, where we assume that $|\mathcal{V}|=O(M)$ to make the complexity of selecting $v_\mathrm{max}$ in Algorithm \ref{Algo5} negligible.

\begin{table*}
\caption{Comparison of algorithms for Problem 1}\label{Tab2}
\centering
\begin{tabular}{@{} c || c | c @{}}
\cline{1-3}
Algorithm ($^*$modified considering the battery constraint) & Complexity & Performance gap \\ \cline{1-3}
Exhaustive search$^*$ \cite{Zhang:2019} & $O(M!M^{3.5}N^2)$ & 0\\ \cline{1-3}
Exhaustive search with fixed association$^*$ \cite{Zhang:2019} & $O(M^{3.5}N^2)$ & $O(MNd_0/v_q+N\tau_\mathrm{max})$ \\ \cline{1-3}
Exhaustive search with quantization$^*$ \cite{Zhang:2019} & {$O(M^4Q^2N^2)$} & $O(MNd_0/v_q+N\tau_\mathrm{max})$ \\ \cline{1-3}
Intersection method$^*$ \cite{Chen:2020} by checking outages via Algorithm \ref{Algo3} & $O(M^4N^2)$ & $O(MNd_0/v_q+N\tau_\mathrm{max})$ \\ \cline{1-3}
Ours (Generalized intersection method with battery constraint)  & $O(M^6)$ & $0$ \\ \cline{1-3}
\end{tabular}
\end{table*}

\begin{theorem}\label{Thm4} 
The GIM-B algorithm outputs an optimal solution for Problem 1 if the power consumption model $P_\mathrm{UAV}(v)$ is convex in the range of the UAV speed.
\end{theorem}
\begin{proof}
It is immediate from Theorems \ref{Thm1} and \ref{Thm3} and the optimality of the Dijkstra algorithm. More specifically, 
\begin{enumerate}
    \item Theorem \ref{Thm1} means that every path between two CSs at the local level has the minimum travel distance.
    \item Theorem \ref{Thm3} implies that flying with the same speed in each path at the local level is optimal. Hence, the GIM-B algorithm derives the minimum travel time for the paths.
    \item Under the graph $\overrightarrow{G}_\mathrm{GL}$ with the minimized edge weights, an optimal trajectory from $\mathbf{u}_0$ to $\mathbf{u}_F$ at the global level is derived by applying the Dijkstra algorithm.
\end{enumerate}
\end{proof}

\begin{theorem}\label{Thm5} 
If the number of CSs is smaller than or equal to the number of BSs, i.e., $N\leq M$, then the time complexity of the GIM-B algorithm is $O(M^6)$. 
\end{theorem}
\begin{proof}
Let us first state the cardinalities of the following sets: $|V_\mathrm{all}|=O(M^2)$, $|V_\mathrm{LO}|=O(M^2)$, and $|V_\mathrm{GL}|=O(N)$. The steps of Algorithm \ref{Algo4} have the following complexities:
\begin{itemize}
    \item Step 1. Outage test: For a line segment, performing the function ChkOut and selecting a memory to save the line segment among $E_\mathrm{in},E_1,...,E_{N+2}$ have the complexities $O(M^2)$ and $O(N)$, respectively. Since each line segment $\overline{\mathbf{x}_1\mathbf{x}_2}$ for $\mathbf{x}_1,\mathbf{x}_2\in V_\mathrm{all}$ and $\mathbf{x}_1\neq\mathbf{x}_2$ should be checked whether experiencing an outage, the complexity of this step is $(O(M^2)+O(N))\cdot |V_\mathrm{all}|^2=O(M^6)$.
    \item  Step 2. Local level search: The complexity of deriving an optimal path between a pair of CSs at the local level can be proved similarly with the proof of Theorem \ref{Thm2}. However, this algorithm constructs the edge set $E_\mathrm{LO}$ with only complexity $O(N)$ by just loading some of the saved memories $E_\mathrm{in},E_1,...,E_{N+2}$. Hence, the complexity of deriving an optimal path in the local level is $O(N)$+$O(M^4)=O(M^4)$. Since there are $O(N^2)$ pairs of the CSs, the complexity of the step is $O(M^4N^2)$.
    \item Step 3. Global level search: This complexity is dominated by applying the Dijkstra algorithm at the graph $\overrightarrow{G}_\mathrm{GL}$ with the complexity $O(|V_\mathrm{GL}|^2)=O(N^2)$ \cite{West:2001}. 
\end{itemize}
Consequently, the complexity of the GIM-B algorithm is $O(M^6)$ for $N\leq M$.
\end{proof}
We note that $N\leq M$ in general because CSs are more expensive and sparse than BSs. Even though incorporating the battery limit, our algorithm maintains the same complexity order as the generalized intersection method if $N\leq M$. 

Table \ref{Tab2} presents a comparison of our algorithm with existing algorithms in \cite{Zhang:2019,Chen:2020} addressing Problem 1. We point out that the existing algorithms in Table \ref{Tab2} retain the same name as Table \ref{Tab1} despite slight modifications to accommodate the battery limit. In more detail, we adapt the existing algorithms following a similar process to Algorithm \ref{Algo4}: Step 1 is omitted because this step is impossible to apply at every algorithm in \cite{Zhang:2019} and does not benefit for the intersection method \cite{Chen:2020}. Step 2 implements the algorithms with minor adjustments through regarding each CS pair as the initial and final points and testing whether the UAV can travel the outputted path in the local level within the battery limit. Step 3 uses the Dijkstra algorithm to find the trajectory in the global level. To perform meaningful comparisons with the results in Table \ref{Tab1}, it is assumed that the UAV's flying speed is fixed at $v_q$ (i.e., $\mathcal{V}=\{0,v_q\}$) to effectively analyze the performance gap in Table \ref{Tab2}, and there always exists a feasible trajectory between the initial and final points in Step 3 for every algorithm (i.e.,  $h_\mathrm{Gfea}=1$). The key insights for Table \ref{Tab2} are outlined as follows:
\begin{itemize}
    \item Our approach yields an optimal solution for Problem 1 in polynomial time. 
    \item For the sub-optimal algorithms, their performance gaps increase with the increase of $N$ by the cumulative effect of the gaps encountered in the path planning between each CS pair. These gaps also vary according to the maximum delay $\tau_\mathrm{max}$ for battery replacement since the UAV may visits more CSs in the paths via sub-optimal algorithms.
    \item In contrast to Table \ref{Tab1}, for ES-Q algorithm, its performance gap does not diminish in the number $Q$ of quantizations, since applying the ES-Q algorithm in Step 2 may lead to the disappearance of few edges in the graph $\overrightarrow{G}_\mathrm{GL}$ at the global level of our GIM-B algorithm by the battery limit.
\end{itemize}
The aforementioned analysis implies that our intersection point-based algorithms have more advantages compared to the benchmark algorithms in the presence of the battery constraint and the CSs.

\begin{remark}\label{rmk3}
When $H_\mathrm{CS}<H$, we can solve Problem 1 by including the take-off and the landing times at charging station $C_n$ in overall delay $\tau_{C_n}$ for $n\in[1:N]$ and considering the consumed energy for them in the battery capacity model \eqref{eq:2}. Similarly, we can check that our GIM-B algorithm is applicable in the case that the altitude of the initial or the final point is lower than $H$ with slight modification.
\end{remark}

\section{Other Objectives}\label{sec5}
In this section, we introduce two UAV path planning problems to move from the initial point $\mathbf{u}_0$ to the final point $\mathbf{u}_F$ under the connectivity and battery constraints, similar to Problem 1 but with different objectives. We consider the objectives of minimizing the energy consumption of the UAV and of maximizing the payload weight that can be delivered in Sections \ref{sec5A} and \ref{sec5B}, respectively.
\subsection{Minimum Energy Consumption}\label{sec5A}
This subsection focuses on finding the UAV trajectory to minimize the energy consumption of the UAV, considering the energy efficiency as in \cite{Zeng:2019,Qi:2020}. We can formulate the optimization problem as follows:
\begin{align} 
&\textbf{Problem 2} \cr
&\mbox{Objective:~}~ \min_{T\geq 0, \{\mathbf{u}(t),\ \psi(t),\ t\in[0,T]\}} \int_{t=0}^{T} {P_\mathrm{UAV}(v(t)) \over \eta }\mathrm{d}t\label{eq:eff1}\\
&\mbox{Constraints: }\cr
&\text{\eqref{eq:9}-\eqref{eq:17}}, \label{eq:eff2}
\end{align}where the objective \eqref{eq:eff1} is the total propulsion energy consumption of the UAV while traveling from $\mathbf{u}_0$ to $\mathbf{u}_F$.

To solve Problem 2, we introduce a slightly modified version of our GIM-B algorithm, called energy-efficient generalized intersection method with battery constraint (EGIM-B). This algorithm is almost the same as the original GIM-B algorithm, except that the weight of each edge at the global level is now given as the minimum energy consumption to travel between the corresponding two CSs and the flight speed during the travel over the edge is selected to minimize the energy consumption. The following proposition shows that the EGIM-B algorithm yields an optimal solution for Problem 2.
\begin{proposition}\label{pro1}
The EGIM-B algorithm outputs an optimal solution for Problem 2.
\end{proposition}
\begin{proof}
This proof is immediate from Theorem \ref{Thm1} and the optimality of the Dijkstra algorithm because
\begin{enumerate}
    \item To minimize the propulsion energy consumption for traveling between two distinct CSs at the local level, the travel distance between those CSs should be minimized. Theorem \ref{Thm1} implies that all paths at the local level achieve the minimum travel distances.
    \item Under the graph in global level with the minimized edge weights (i.e., minimized energy consumption in local level), a trajectory that minimizes the energy consumption for traveling from $\mathbf{u}_0$ to $\mathbf{u}_F$ is derived via the Dijkstra algorithm.
\end{enumerate} 
\end{proof}
We note that in an optimal UAV trajectory for Problem 2, the UAV's flying speed is fixed to $v_\mathrm{eff}\in\mathcal{V}$, which minimizes the UAV's propulsion energy consumption per unit distance and is represented as follows: $v_\mathrm{eff}=\mathrm{argmin}_{v\in\mathcal{V}\setminus 0} {{P_\mathrm{UAV}(v)/(\eta v)}}$
\sh{The following proposition shows that the EGIM-B algorithm operates with polynomial time complexity under the assumption $|\mathcal{V}|=O(M)$, whose proof is immediate from Theorem \ref{Thm5} and hence omitted.}
\begin{proposition}\label{pro2} 
\sh{ If the number of CSs is smaller than or equal to the number of BSs, i.e., $N\leq M$, then the time complexity of the EGIM-B algorithm is $O(M^6)$. }
\end{proposition}

\subsection{Maximum Deliverable Payload Weight}\label{sec5B}
In this subsection, we focus on deriving the maximum deliverable payload weight, instead of the minimum mission time. We can formulate the optimization problem as
\begin{align} 
&\textbf{Problem 3} \cr
&\mbox{Objective:~}~~~~ \max_{w_3\geq 0, \{\mathbf{u}(t),\ \psi(t),\ t\in[0,T]\}} w_3 \label{eq:33}\\
&\mbox{Constraints: }\cr
&0\leq T<\infty \label{eq:34} \\
&\text{\eqref{eq:9}-\eqref{eq:17}}, \label{eq:35}
\end{align}
where \eqref{eq:34} means that the UAV succeeds to deliver the payload from $\mathbf{u}_0$ to $\mathbf{u}_F$ within a finite time. We note that the propulsion power consumption $P_\mathrm{UAV}(v(t))$ in \eqref{eq:15} depends on the payload weight $w_3$. 

To solve Problem 3, we propose the bottleneck edge search method described in Algorithm \ref{Algo6}.  
\begin{algorithm}[t]
\caption{Bottleneck Edge Search Method} \label{Algo6}
\textbf{Input:} $\mathbf{u}_0$, $\mathbf{u}_F$, $\mathcal{V}$, $\mathbf{c}_n$, $\ell_\mathrm{LO}(\mathbf{c}_n,\mathbf{c}_{n'})$, $h_\mathrm{Lfea}(\mathbf{c}_n,\mathbf{c}_{n'})$, $w_1$, $w_2$, $\epsilon_w$, $k_\mathrm{max}$ for $n\in [1:N+1]$, $n'\in[1:N]\cup \{N+2\}$, $n<n'$
\begin{algorithmic}[1]
\State $V_\mathrm{GL}\leftarrow\{\mathbf{u}_0,\mathbf{u}_F,\mathbf{c}_1,...,\mathbf{c}_N\}$, $E'_\mathrm{GL}\leftarrow \emptyset$, $w_3\leftarrow 0$
\State $\mathbf{c}_{N+1} \leftarrow \mathbf{u}_0$, $\mathbf{c}_{N+2} \leftarrow \mathbf{u}_F$, $h_\mathrm{sp}\leftarrow 1$
\LeftComment{Step 1. Graph construction: Construct a graph $G'_\mathrm{GL}$ whose vertex set consists of CSs and edge set consists of edges between two connected CSs. The weight of each edge is the minimum travel distance between two CSs.}
\For {$n\in [1:N+1]$, $n'\in[1:N]\cup \{N+2\}$, $n<n'$}
    \LeftComment{Parameters $h_\mathrm{Lfea}$ and $\ell_\mathrm{LO}$ are described in Algorithm \ref{Algo4}.}
    \If{$h_\mathrm{Lfea}(\mathbf{c}_n,\mathbf{c}_{n'})=1$}
    \State $E'_\mathrm{GL}\leftarrow E'_\mathrm{GL}\cup (\mathbf{c}_n,\mathbf{c}_{n'},\ell_\mathrm{LO}(\mathbf{c}_n, \mathbf{c}_{n'}))$
    \EndIf
\EndFor
\State $G'_\mathrm{GL}\leftarrow (V_\mathrm{GL}, E'_\mathrm{GL})$ 
\LeftComment{Step 2. Bottleneck edge search: Find the longest connectivity-critical edge in $G'_\mathrm{GL}$.}
\LeftComment{Function \textbf{BFS} is described in line $1$ at Algorithm \ref{Algo4}.}
\State $h'_\mathrm{Gfea}\leftarrow$ \textbf{BFS}$(\mathbf{u}_0,\mathbf{u}_F,G'_\mathrm{GL})$
\If{$h'_\mathrm{Gfea}=1$}
\While{$h'_\mathrm{Gfea}=1$} 
    \State $(\mathbf{c}_k,\mathbf{c}_{k'},\ell_\mathrm{bott}) \leftarrow \!\!\!\!\!\!\!\!\!\!\!\!\!\!\!\!\!\!
    \underset{{(\mathbf{c}_n,\mathbf{c}_{n'},\ell_\mathrm{LO}(\mathbf{c}_n, \mathbf{c}_{n'}))\in E'_\mathrm{GL}}}{\mathrm{argmax}}$\!\!\!\!\!\!\!\!\!\!\!\!\!\!\!\!\! $\ell_\mathrm{LO}(\mathbf{c}_n, \mathbf{c}_{n'})$
    \LeftComment{Eliminate the longest edge in $E'_\mathrm{GL}$.}
    \State $E'_\mathrm{GL}\leftarrow E'_\mathrm{GL}\setminus (\mathbf{c}_k,\mathbf{c}_{k'},\ell_\mathrm{bott})$
    \State $G'_\mathrm{GL}\leftarrow (V_\mathrm{GL}, E'_\mathrm{GL})$
    \State $h'_\mathrm{Gfea}\leftarrow$ \textbf{BFS}$(\mathbf{u}_0,\mathbf{u}_F,G'_\mathrm{GL})$
\EndWhile
\Else
    \State $\ell_\mathrm{bott}\leftarrow \infty$
\EndIf
\LeftComment{Step 3. Weight search: Derive the maximum deliverable payload weight $w_3$.}
\LeftComment{Parameter $\epsilon_w>0$ is a sufficiently small constant.}
\While{$h_\mathrm{sp}=1$, $w_3\leq k_\mathrm{max}\epsilon_w$} \hfill
    \State $w_3\leftarrow w_3+\epsilon_w$
    \LeftComment{Function \textbf{ChkSp} is described in Algorithm \ref{Algo5}.}
    \State $(h_\mathrm{sp},v)\leftarrow$ \textbf{ChkSp}$(\ell_\mathrm{bott},\mathcal{V},w_1+w_2+w_3,w_2)$
\EndWhile
\State $w_3\leftarrow w_3-\epsilon_w$
\end{algorithmic}
\textbf{Output:} $w_3$
\end{algorithm}
This algorithm initially sets zero payload weight, i.e., $w_3=0$. It first constructs an undirected weighted graph $G'_\mathrm{GL}$ whose vertex set consists of CSs (by
treating the initial and the final points also as CSs) and whose edge set includes an edge between two CSs only when there exists a path between the two CSs $(h_\mathrm{Lfea}=1)$, with the weight of the travel distance $\ell_\mathrm{LO}$ (in lines $3$-$8$). Note that the parameters $h_\mathrm{Lfea}$ and $\ell_\mathrm{LO}$ for each pair of CSs can be obtained by Algorithm \ref{Algo4}. After constructing the graph, it finds the bottleneck edge, which is the longest connectivity-critical edge in the graph (in lines $9$-$19$). To this end, the algorithm first checks whether $\mathbf{u}_0$ and $\mathbf{u}_F$ are connected by applying the function BFS \cite{West:2001}. If connected, it repeatedly eliminates the longest edge from the graph and then checks whether they are connected in the graph until not connected ($h'_\mathrm{Gfea}=0$). After the repetition ends, the most recently deleted edge is set as the bottleneck edge, with the edge weight $\ell_\mathrm{bott}$. Finally, the maximum deliverable payload weight over the graph is derived (in lines $20$-$24$). It first checks whether the UAV can travel the distance $\ell_\mathrm{bott}$ without battery replacement ($h_\mathrm{sp}=1$) or not ($h_\mathrm{sp}=0$) at $w_3=0$ via the function ChkSp whose pseudo code is in algorithm \ref{Algo5}. Then, it iterates this process while increasing $w_3$ in sufficiently small increments $\epsilon_w>0$ until the UAV cannot deliver the payload over the bottleneck edge or $w_3$ exceeds the limit $k_\mathrm{max}\epsilon_w$ of the payload weight. This algorithm outputs the maximum payload weight $w_3\in\{0,\epsilon_w,...,k_\mathrm{max}\epsilon_w\}$ which can be delivered over the bottleneck edge.

\sh{Now, the following theorems show that our bottleneck edge search method yields an optimal solution of Problem 3 in polynomial time.}
\begin{theorem}\label{Thm6} 
Assume that the payload weight $w_3$ is selected from the set $\{0,\epsilon_w,...,k_\mathrm{max}\epsilon_w\}$ for $\epsilon_w>0$ and $k_\mathrm{max}\in\mathbb{N}$. Then, the bottleneck edge search method outputs the optimal solution for Problem 3 if the power consumption model $P_\mathrm{UAV}(v)$ is convex in the range of the UAV speed.  
\end{theorem}
\begin{proof} 
If the UAV can travel the bottleneck edge without battery replacement, then the payload can be delivered from $\mathbf{u}_0$ to $\mathbf{u}_F$ since it can be also delivered over an edge shorter than the bottleneck edge under the battery constraint. Hence, it is sufficient only to consider whether the payload can be delivered over the bottleneck edge. For finding the maximum deliverable payload weight, we note that it is sufficient only to consider a fixed speed while traveling the bottleneck edge, as justified in Theorem \ref{Thm3}. Consequently, our method yields the optimal solution for Problem 3.
\end{proof}
\sh{
\begin{theorem}\label{Thm7} 
Assume that the payload weight $w_3$ is selected from the set $\{0,\epsilon_w,...,k_\mathrm{max}\epsilon_w\}$ for $\epsilon_w>0$ and $k_\mathrm{max}\in\mathbb{N}$. Then, the 
time complexity of the bottleneck edge search method is $O(N^4+k_\mathrm{max}|\mathcal{V}|)$.
\end{theorem}
\begin{proof}
Note that $|V_\mathrm{GL}|=O(N)$ and $|E'_\mathrm{GL}|=O(N^2)$. The steps in Algorithm \ref{Algo6} have the following complexities:
\begin{itemize}
\item Step 1. Graph construction: This step has the complexity $O(N^2)$ since it just loads all the feasible edges (i.e., $h_\mathrm{Lfea}=1$) between two vertices.
\item Step 2. Bottleneck edge search: First, the complexity of finding an edge to be eliminated from graph $G'_\mathrm{GL}$, which is a longest edge in the graph, is $O(|E'_\mathrm{GL}|)=O(N^2)$. Next, the complexity of applying function BFS in graph $G'_\mathrm{GL}$ is $O(|V_\mathrm{GL}|^2)=O(N^2)$ \cite{West:2001}. Since such process is repeated at most $|E'_\mathrm{GL}|$ times, the complexity of searching a bottleneck edge is $(O(N^2)+O(N^2))\cdot |E'_\mathrm{GL}|=O(N^4)$.
\item Step 3. Weight search: The complexity of checking whether each payload $w_3\in\{0,\epsilon_w,...,k_\mathrm{max}\epsilon_w\}$ is deliverable via the function ChkSp is $(k_\mathrm{max}+1)\cdot O(|\mathcal{V}|)=O(k_\mathrm{max}|\mathcal{V}|)$.
\end{itemize}
Consequently, the complexity of the bottleneck edge search method is $O(N^4+k_\mathrm{max}|\mathcal{V}|)$.
\end{proof}
}

\section{Numerical Results}\label{sec6}

This section provides a range of numerical results to assess the effectiveness of the GIM-B algorithm and the bottleneck edge search method. It is assumed that $M=19$ BSs and $N=5$ CSs are deployed across a $10\mathrm{km}\times 10\mathrm{km}$ region wherein the UAV flies from $\mathbf{u}_0=(-300\mathrm{m},300\mathrm{m})$ to $\mathbf{u}_F=(6650\mathrm{m},7900\mathrm{m})$ at the same altitude $H=100\mathrm{m}$.
The coverage radius of $\mathrm{BS}_m$ is established by $d_0=1484.6\mathrm{m}$ and $\lambda_m\in[0,800]\mathrm{m}$ for $m\in[1:M]$. Note that the base coverage radius $d_0$ can be obtained via the communication model in \eqref{eq:3.1}-\eqref{eq:6}, whose parameters are described in Table \ref{Tabcomm}.  
\begin{table}
\caption{Parameters for communication model}\label{Tabcomm}
\centering
\begin{tabular}{@{} c || c | c @{}}
\cline{1-3}
Symbol &  Definition &  Value\\ \cline{1-3}
$H_\mathrm{BS}$ & Height of BS & $35\mathrm{m}$\\ \cline{1-3}
$\mathrm{SINR_\mathrm{th}}$ & Hard SINR threshold & $12\mathrm{dB}$\\ \cline{1-3}
$\mathrm{SNR_\mathrm{ref}}$ & SNR at distance $1\mathrm{m}$ in free space & $95\mathrm{dB}$\\ \cline{1-3}
$\mu_1$ & Parameter for LoS probability & $4.880$\\ \cline{1-3}
$\mu_2$ & Parameter for LoS probability & $0.429$\\ \cline{1-3}
$\zeta_1$ & Excessive pathloss (LoS) & $0.1\mathrm{dB}$\\ \cline{1-3}
$\zeta_2$ & Excessive pathloss (NLoS) & $21\mathrm{dB}$\\ \cline{1-3}
\end{tabular}
\end{table}
\sh{The total delay for battery replacement at each CS is set to 100s, i.e., $\tau_{C_n}=100\mathrm{s}$ for $n\in[1:N]$, by considering approximately 60 seconds for the replacement itself \cite{Lee:2015} and about 40 seconds of excessive penalty factors, but we note that  our proposed algorithms also operate when the delay at each CS is not the same.} 
 The UAV can adjust its speed $v$ in the set $\mathcal{V}=[0:1:30]\mathrm{m/s}$. The total UAV weight, inclusive of its payload, is assumed to be $w=2.97\mathrm{kg}$, where $w_1=1.07\mathrm{kg}$, $w_2=0.9\mathrm{kg}$, and $w_3=1\mathrm{kg}$. In the propulsion power consumption model \eqref{eq:1}, $P_1$, $P_2(w)$, and the mean rotor induced speed for hovering $v_0(w)$  are given as the following:
\begin{align}
P_1&={(\delta_p\rho/ 8)} (N_rN_bL_cR_r) v_\mathrm{tip}^3,\label{eq:s1}\\
P_2(w)&=(1+k_\mathrm{cf}){(wg)^{3/2}/\sqrt{2\rho N_r\pi R_r^2}},\label{eq:s2}\\
v_0(w)&=\sqrt{{wg}/({2\rho N_r\pi R_r^2})},\label{eq:s3}
\end{align}
where the parameters in \eqref{eq:s1}-\eqref{eq:s3} are described in Table \ref{Tab3}. 
For simulations, our choice of parameter values for the communication model \eqref{eq:3.1}-\eqref{eq:6}, and the power consumption model \eqref{eq:s1}-\eqref{eq:s3} and battery model \eqref{eq:2}-\eqref{eq:3} are summarized in Tables \ref{Tabcomm} and \ref{Tab3}, respectively.\footnote{We refered to the communication parameters for suburban environment in \cite{Al-Hourani:2014,Al-Hourani:2014_2,TR:2018} and the power consumption and battery parameters in \cite{Zhang:2021_2,Zeng:2019}.}
\begin{table}
\caption{Parameters for power consumption and battery model}\label{Tab3}
\centering
\begin{tabular}{@{} c || c | c @{}}
\cline{1-3}
Symbol &  Definition &  Value\\ \cline{1-3}
$\delta_p$ & Profile drag coefficient & $0.012$\\ \cline{1-3}
$N_r$ & Number of rotors (quadcopter) & $4$\\ \cline{1-3}
$N_b$ & Number of blades per rotor & $4$\\ \cline{1-3}
$L_c$ & Blade chord length & $0.0157\mathrm{m}$\\ \cline{1-3}
$R_r$ & Rotor radius & $0.07\mathrm{m}$\\ \cline{1-3}
$v_\mathrm{tip}$ & Tip speed of a blade & $14\mathrm{m/s}$\\ \cline{1-3}
$k_\mathrm{cf}$ & Incremental correlation factor & $0.1$\\ \cline{1-3}
$S_\mathrm{FP}$ & Fuselage equivalent flat area & $0.03\mathrm{m^2}$\\ \cline{1-3}
$\rho$ & Air density & $1.225\mathrm{kg/m^3}$\\ \cline{1-3}
$g$ & Gravitational acceleration & $9.807\mathrm{m/s^2}$\\ \cline{1-3}
$\epsilon_\mathrm{batt}$ & Energy density of battery per kg & $540\mathrm{kJ/kg}$\\ \cline{1-3}
$\gamma$ & Depth of discharge & $0.7$\\ \cline{1-3}
$\eta$ & Ratio of transferable energy & $0.7$\\ \cline{1-3}
$r_\mathrm{safe}$ & Energy reserving factor & $1.2$\\ \cline{1-3}
\end{tabular}
\end{table}

Fig. \ref{Figs2} illustrates the flight trajectory and the corresponding mission time $T$ for both the GIM-B and existing algorithms in \cite{Zhang:2019,Chen:2020}.
\begin{figure}
\centering
\includegraphics[width=0.8\columnwidth]{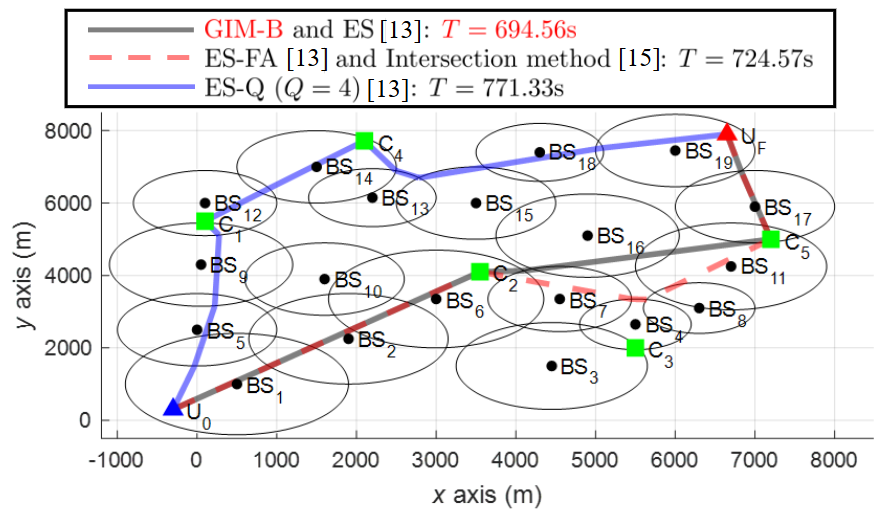}
\caption{Comparison of trajectory and the corresponding mission time $T$ for algorithms to solve Problem 1.}\label{Figs2}
\end{figure} 
The ES-FA algorithm \cite{Zhang:2019} and intersection method \cite{Chen:2020} yield the identical trajectory, which differs from the optimal result of our approach and the ES algorithm \cite{Zhang:2019}. The ES-Q algorithm \cite{Zhang:2019} at $Q=2$ fails to output any UAV trajectory owing to the battery limit. At $Q=4$, it exhibits a higher time complexity compared to our algorithm since $QN>M$. Nonetheless, it achieves a considerably longer mission time $T$ even than the other sub-optimal algorithms, because it fails to derive a path between $\mathbf{u}_0$ and $\mathbf{c}_2$ under the quantization points that can be traveled within the battery capacity.
Tables \ref{TabTime} and \ref{TabEnergy} show the excessive  mission time and the excessive propulsion energy consumption of the existing algorithms \cite{Zhang:2019,Chen:2020} compared to those of GIM-B algorithm,  respectively,  over 100 distinct maps. In the process of constructing each map, we first randomly select the locations of the initial and final points, BSs, and CSs according to the 2D uniform distribution over the $10\mathrm{km}\times 10\mathrm{km}$ region, and the coverage offsets of BSs according to the uniform distribution over the range $[0.800]\mathrm{m}$. Then, we accept it as a feasible map only if there exists a feasible trajectory between the initial and the final points under the connectivity and battery constraints. In Table \ref{TabTime}, the ES-FA algorithm \cite{Zhang:2019} and the intersection method \cite{Chen:2020} achieve the optimal travel time for more than half of the maps, because they can search an optimal trajectory if they succeed in finding an optimal BS association sequence. The ES-Q algorithm \cite{Zhang:2019} with $Q=4$ does not achieve the optimal travel time for all the maps, as it restricts the trajectories to lie on a weighted graph constructed based on the quantization points. We note that for a few maps, some suboptimal algorithms fail to yield any feasible UAV trajectory. Such infeasible cases occur due to the battery constraint, i.e., for the set of breakpoints that the corresponding algorithm considers, there is no path that the UAV can travel under the battery constraint. In Table \ref{TabEnergy}, we can see that some suboptimal algorithms achieve lower energy consumption compared to that of GIM-B algorithm for a few maps, since the GIM-B algorithm aims to minimize the mission time. It is desirable to use the EGIM-B algorithm in Section \ref{sec5A} if the focus is on the energy efficiency.  
\begin{table}
\caption{Excessive mission time compared to that of GIM-B algorithm over 100 random maps}\label{TabTime}
\centering
\begin{tabular}{@{} c || c | c @{}}
\cline{1-3}
\multirow{2}{2.0cm}{\centering Excessive mission time  [\%]} & \multicolumn{2}{c}{Number of maps}\\ \cline{2-3}
 & \makecell{ES-FA \cite{Zhang:2019} \\ Intersection method \cite{Chen:2020}} & \makecell{ES-Q \cite{Zhang:2019}\\ $(Q=4)$ }\\ \cline{1-3}
$0\%$ & 66 & 0 \\ \cline{1-3}
$(0,1]\%$ & 21 & 85 \\ \cline{1-3}
$(1,5]\%$ & 10 & 13 \\ \cline{1-3}
$(5,10]\%$ & 3 & 0 \\ \cline{1-3}
Infeasible & 1 & 2 \\ \cline{1-3}
\end{tabular}
\end{table}
\begin{table}
\caption{Excessive energy consumption compared to that of GIM-B algorithm over 100 random  maps}\label{TabEnergy}
\centering
\begin{tabular}{@{} c || c | c @{}}
\cline{1-3}
\multirow{2}{2.5cm}{\centering Excessive energy consumption [\%]} & \multicolumn{2}{c}{Number of maps}\\ \cline{2-3}
 & \makecell{ES-FA \cite{Zhang:2019} \\ Intersection method \cite{Chen:2020}} & \makecell{ES-Q \cite{Zhang:2019}\\ $(Q=4)$ }\\ \cline{1-3}
$(-20,-1]\%$ & 1 & 2 \\ \cline{1-3}
$(-1,0)\%$ & 4 & 9 \\ \cline{1-3}
$0\%$ & 65 & 0 \\ \cline{1-3}
$(0,1]\%$ & 21 & 87 \\ \cline{1-3}
$(1,5]\%$ & 8 & 0 \\ \cline{1-3}
Infeasible & 1 & 2 \\ \cline{1-3}
\end{tabular}
\end{table}
Fig. \ref{Figs3} shows the optimal graph at the global level and the corresponding maximum possible speed $v_\mathrm{max}$ for each edge under the environment in Fig. \ref{Figs2}. We can see that for each edge, the maximum travel speed $v_\mathrm{max}$ decreases as its travel distance increases.
\begin{figure}
\centering
\includegraphics[width=0.8\columnwidth]{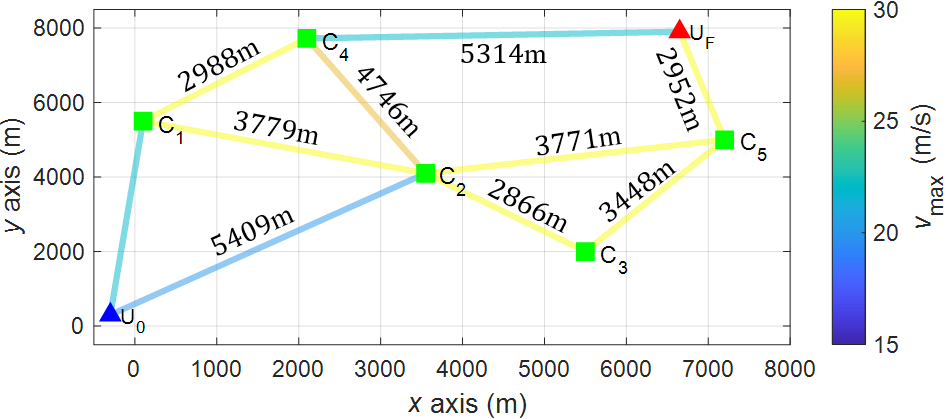}
\caption{Optimal graph at the global level and the corresponding maximum possible speed $v_\mathrm{max}$ for each edge.}\label{Figs3}
\end{figure}

Fig. \ref{Figs4} compares the optimal UAV trajectory and the corresponding mission time $T$ for different payload weight $w_3$, battery weight $w_2$, and delay $\tau_{C_1}$ at charging station $C_1$, where the locations of CSs are changed from Fig. \ref{Figs2}.
We can see that the UAV avoids $C_1$ with the higher delay $\tau_{C_1}=200\mathrm{s}$ for battery replacement (red), it visits more CSs with the larger payload weight $w_3=1.5\mathrm{kg}$ (blue), and it visits less CSs with the larger battery weight $w_2=1.2\mathrm{kg}$ (green).
\begin{figure}
\centering
\includegraphics[width=0.8\columnwidth]{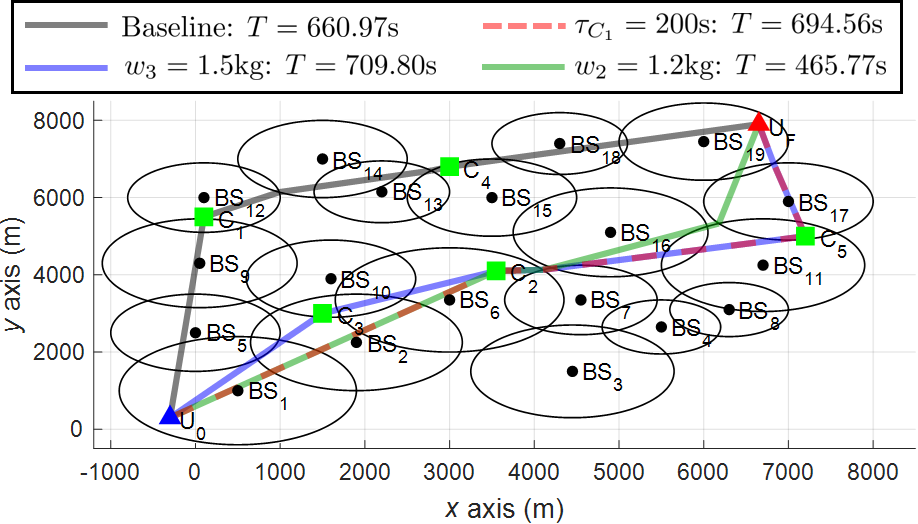}
\caption{Comparison of optimal trajectory and the corresponding mission time $T$ for the cases that delay $\tau_{C_1}=200\mathrm{s}$ for battery replacement, payload weight $w_3=1.5\mathrm{kg}$, and battery weight $w_2=1.2\mathrm{kg}$.}\label{Figs4}
\end{figure}
Fig. \ref{Figs5} plots the optimal mission time $T$ across different delays for battery replacement and payload weights $w_3\in[0:0.1:3.5]\mathrm{kg}$ under the same environment as in Fig. \ref{Figs4}. We can verify that $T$ increases as the delay for battery replacement and the payload weight increase and that the payload cannot be delivered from $\mathbf{u}_0$ to $\mathbf{u}_F$ if $w_3$ is too large, i.e., if it exceeds $2.8\mathrm{kg}$ under this setting.
\begin{figure}
\centering
\includegraphics[width=0.7\columnwidth]{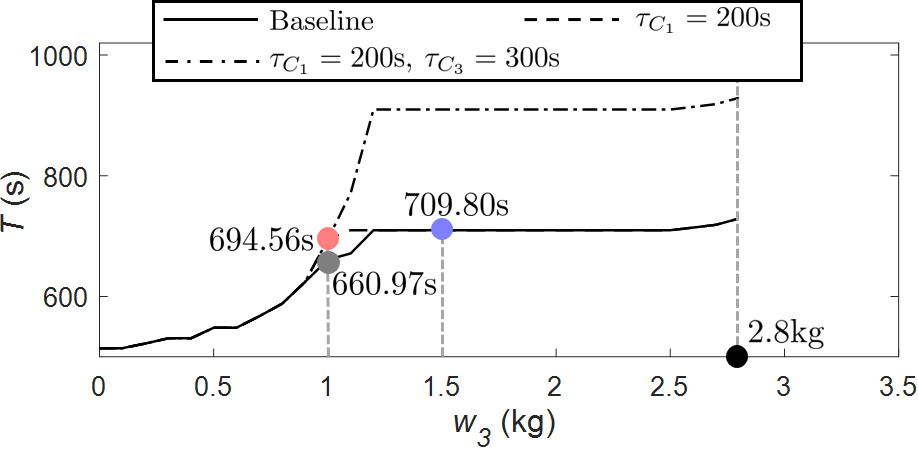}
\caption{Payload weight $w_3$ versus optimal mission time $T$ for different $\tau_{C_1}$ and $\tau_{C_3}$.}\label{Figs5}
\end{figure}

Fig. \ref{Figs7} compares the optimal mission time $T$ for the case that the UAV can  fly with a fixed speed of $v_\mathrm{fix}\in [15:1:30]\mathrm{m/s}$ (fixed speed) and for the case that it can change its speed in the speed set $\mathcal{V}$ (dynamic speed) under the same environment as in Fig. \ref{Figs4}. Note that in the fixed speed case, the UAV chooses its speed in the set $\{0,v_\mathrm{fix}\}$.  We can check that the dynamic speed case has a lower travel time than the fixed speed case for every $v_\mathrm{fix}\in [15:1:30]\mathrm{m/s}$ because the maximum allowable speed between each pair of CSs at the local level depends on its travel distance as shown in Fig. \ref{Figs3}. In small battery weight $w_2=0.6\mathrm{kg}$, any trajectory from $u_0$ to $u_F$ cannot be derived in the fixed speed case with $v_\mathrm{fix}>22\mathrm{m/s}$ since flying at a high speed is not efficient in terms of the energy consumption. 
\begin{figure}
\centering
\includegraphics[width=0.7\columnwidth]{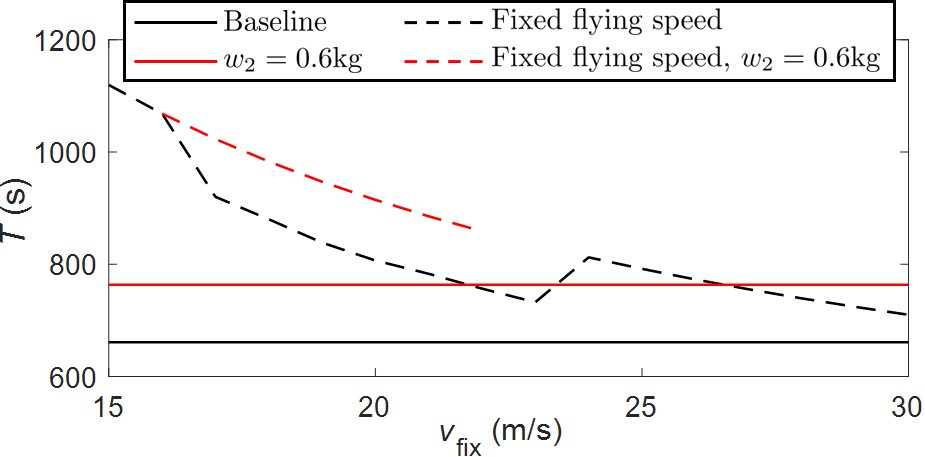}
\caption{Flying speed $v_\mathrm{fix}$ versus optimal mission time $T$ for the cases that the UAV can only fly with a fixed speed $v_\mathrm{fix}\in [15:1:30]\mathrm{m/s}$ and change its speed in the speed set $\mathcal{V}$.}\label{Figs7}
\end{figure}

Fig. \ref{Figs7.1} illustrates the UAV trajectory and the corresponding mission time $T$ and propulsion energy consumption $E_\mathrm{tot}$ for the GIM-B and EGIM-B algorithms, where the locations of CSs are changed from Fig. \ref{Figs2}. We can see that for the GIM-B algorithm, the UAV flies at the maximum possible speed under the battery constraint and frequently visits CSs to increase its flight speed thereby decreasing the mission time, with the sacrifice of the travel distance and the energy consumption. On the other hand, for the EGIM-B algorithm, the UAV always flies at the fixed speed $v_\mathrm{eff}=20\mathrm{m/s}$ which minimizes its energy consumption per unit distance and   minimizes the travel distance to reduce the propulsion energy consumption, with the sacrifice of the mission time. To evaluate the average performance of the GIM-B and EGIM-B algorithms, for the same 100 maps used for Tables \ref{TabTime} and \ref{TabEnergy}, we calculate the average relative performance. On average, the mission time and the energy consumption from the EGIM-B algorithm are $28.39\%$ longer and $14.10\%$ lower than those from the GIM-B algorithm, respectively. 
\begin{figure}
\centering
\includegraphics[width=0.8\columnwidth]{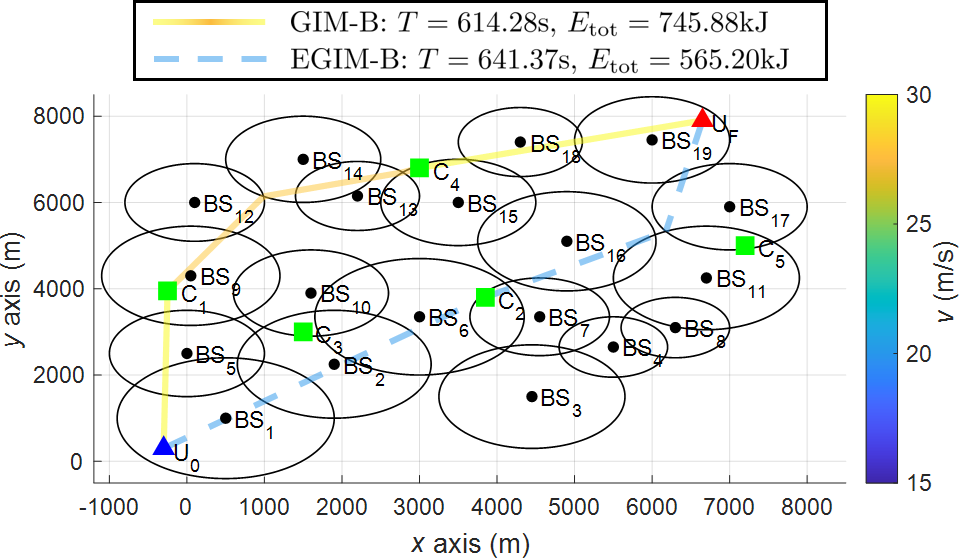}
\caption{Comparison of UAV trajectory and the corresponding mission time $T$ and propulsion energy consumption $E_\mathrm{tot}$ for GIM-B and EGIM-B algorithms.}\label{Figs7.1}
\end{figure}
Fig. \ref{Figs8} plots the maximum deliverable weight $w_3$ according to the battery weights $w_2\in[0.5:0.02:1]\mathrm{kg}$ for different unavailable CSs under the same environment as in Fig. \ref{Figs4}, where we say that charging station $C_n$ is unavailable if its delay for battery replacement is $\tau_{C_n}=\infty$.
We can check that the maximum deliverable weight decreases as $w_2$ decreases and the number of unavailable CSs increases.
\begin{figure}
\centering
\includegraphics[width=0.7\columnwidth]{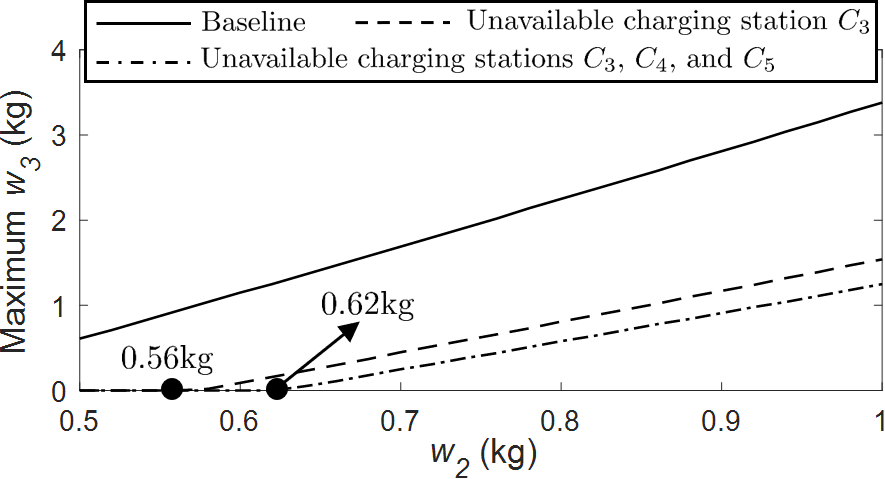}
\caption{Battery weight $w_2$ versus maximum deliverable payload $w_3$ for different unavailable CSs.}\label{Figs8}
\end{figure}

\sh{Finally, let us compare the actual execution time of our GIM-B algorithm with existing graph theory-based algorithms: ES-Q algorithm with $Q=4$ \cite{Zhang:2019} and intersection method \cite{Chen:2020}, for the same 100 maps used for Tables \ref{TabTime} and \ref{TabEnergy}. On average\footnote{\sh{These three algorithms were executed by MATLAB R2024a using Intel Core i5-8250U  and 8GB DDR4-2400 RAM.}}, the execution time of ours is $2.45\mathrm{s}$, which is smaller than $7.21\mathrm{s}$ of the ES-Q algorithm  and $4.46\mathrm{s}$ of the intersection method. This is due to the following reasons. First, the ES-Q algorithm exhibits a higher complexity than ours when $QN>M$. Next, the intersection method performs outage tests separately for each pair of CSs, and this can cause a larger increase in complexity compared to considering all BS associations in our algorithm, especially when there are not many other BSs whose coverage regions overlap with the coverage region of each BS. The complexity analysis in Table \ref{Tab2} is based on worst-case assumptions, and under the aforementioned sparse BS scenarios, our GIM-B algorithm has the complexity order of $\Theta(M^4)$. 
} 


\section{Conclusion}\label{sec7}

For the cellular-enabled UAV's path planning problem in both connectivity and battery constraints, we proposed the generalized intersection method with battery constraint (GIM-B) algorithm that outputs an optimal UAV trajectory in polynomial time. The efficacy of our algorithm, considering both mission completion time and computational complexity, was shown through comprehensive comparisons with existing algorithms both in analytically and numerically.
Furthermore, we proposed the energy-efficient generalized intersection method with battery constraint (EGIM-B) algorithm and the bottleneck edge search method that find the minimum UAV energy consumption and the maximum deliverable payload weight, respectively, under the connectivity and battery constraints.
Various numerical results were provided to show an optimal UAV trajectory and  various performance metrics according to  objectives and environmental parameters.

Let us conclude with some remarks on further works. We assumed that the delay at each CS is fixed over time, but in general, it  changes over time in practice. One interesting approach would be to assume the scenario of time-dependent delays to swap the battery at charging stations and develop shortest path finding algorithms over time-dependent graphs \cite{Orda:1990}. Another interesting scenario would be to consider more practical communication environments based on radio map taking into account signal blockage and reflection by buildings and interference from other BSs \cite{Chen:2017,Zhang:2021}.

\ifCLASSOPTIONcaptionsoff
  \newpage
\fi

\bibliographystyle{IEEEtran}
\bibliography{ref}
\begin{IEEEbiography}[{\includegraphics[width=1in,height=1.25in,clip,keepaspectratio]{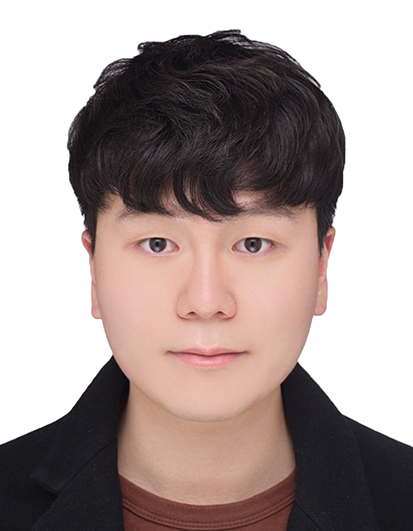}}]
{Hyeon-Seong IM} received the B.S. (summa cum laude) and M.S. degrees from the Pohang University of Science and Technology (POSTECH), Pohang, South Korea, in 2019 and 2020, respectively, and the Ph.D. degree from the School of Electrical Engineering, Korea Advanced Institute of Science and Technology (KAIST), Daejeon, South Korea, in 2024. He is currently a Senior Researcher in LIG Nex1, Seongnam, South Korea. His research interests include military communications, non-terrestrial networks (NTN), and anti-jamming.
\end{IEEEbiography}
\begin{IEEEbiography}[{\includegraphics[width=1in,height=1.25in,clip,keepaspectratio]{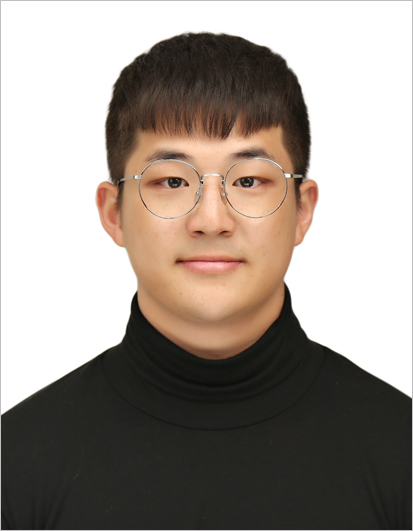}}]
{Kyu-Yeong Kim} (Graduate Student Member, IEEE) received the B.S. (Great Honors) degree in electrical engineering from Korea University in 2022 and the M.S. degree from the School of Electrical Engineering, Korea Advanced Institute
of Science and Technology (KAIST), South Korea, in 2024, where he is currently pursuing the Ph.D. degree.
\end{IEEEbiography}
\begin{IEEEbiography}[{\includegraphics[width=1in,height=1.25in,clip,keepaspectratio]{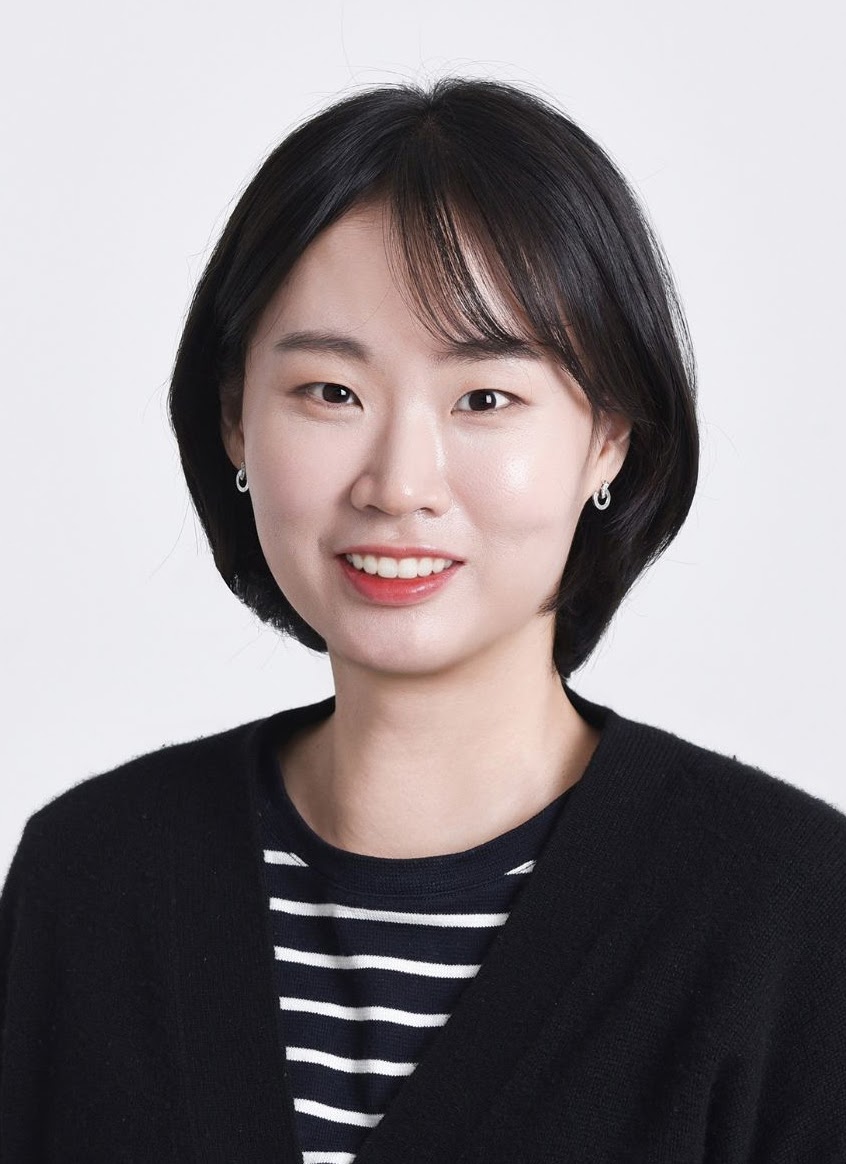}}]
{Si-Hyeon Lee} (Senior Member, IEEE) received the B.S. (summa cum laude) and Ph.D. degrees in electrical engineering from the Korea Advanced Institute of Science and Technology (KAIST), Daejeon, South Korea, in 2007 and 2013, respectively.
She is currently an Associate Professor with the School of Electrical Engineering, KAIST.
She was a Postdoctoral Fellow with the Department of Electrical and Computer Engineering, University of Toronto, Toronto, Canada, from 2014 to 2016, and an Assistant Professor with the Department of Electrical Engineering, Pohang University of Science and Technology (POSTECH), Pohang, South Korea, from 2017 to 2020. Her research interests include information theory, wireless communications, statistical inference, and machine learning. She was a TPC Co-Chair of IEEE Information Theory Workshop 2024. She is currently an IEEE Information Theory Society Distinguished Lecturer (2024-2025) and an Associate Editor for IEEE Transactions on Information Theory. 
\end{IEEEbiography}
\vfill

\end{document}